\documentclass[12pt,reqno]{article}%
\usepackage{graphicx}
\usepackage{amsmath}
\usepackage{amsfonts}
\usepackage{amssymb}
\usepackage{algorithm}
\usepackage{algorithmic}%
\providecommand{\U}[1]{\protect\rule{.1in}{.1in}}
\newtheorem{theorem}{Theorem}[section]

\newtheorem{definition}[theorem]{Definition}

\newenvironment{proof}[1][Proof]{\noindent\textbf{#1.} }{\ \rule{0.5em}{0.5em}}
\begin{document}

\title{The Role of Visibility \ in Pursuit / Evasion Games}
\author{Athanasios Kehagias, Dieter Mitsche, and Pawe{\l } Pra{\l }at}
\date{\today}
\maketitle

\begin{abstract}
The \emph{cops-and-robber} (CR)\ game has been used in mobile robotics as a
discretized model (played on a graph $G$) of \emph{pursuit/evasion} problems.
The \textquotedblleft classic\textquotedblright\ CR version is a \emph{perfect
information game}: the cops' (pursuer's) location is always known to the
robber (evader) and vice versa. Many variants of the classic game can be
defined: the robber can be \emph{invisible} and also the robber can be either
\emph{adversarial} (tries to avoid capture) or \emph{drunk} (performs a random
walk). Furthermore, the cops and robber can reside in either nodes or edges of
$G$. Several of these variants are relevant as models or robotic pursuit /
evasion. In this paper, we first define carefully several of the variants
mentioned above and related quantities such as the \emph{cop number} and the
\emph{capture time}. Then we introduce and study the \emph{cost of visibility}
(\emph{COV}), a quantitative measure of the increase in difficulty (from the
cops' point of view) when the robber is invisible. In addition to our
theoretical results, we present algorithms which can be used to compute
capture times and COV of graphs which are analytically intractable. Finally,
we present the results of applying these algorithms to the numerical
computation of COV.

\end{abstract}

\section{Introduction\label{sec01}}

\emph{Pursuit / evasion} (PE) and related problems (search, tracking,
surveillance) have been the subject of extensive research in the last fifty
years and much of this research is connected to mobile robotics
\cite{chungsearch}. When the environment is represented by a
graph\footnote{For instance, a floorplan can be modeled as a graph, with nodes
corresponding to rooms and edges corresponding to doors. Similarly, a maze can
be represented by a graph with edges corresponding to tunnels and nodes
corresponding to intersections.}, the original PE\ problem is reduced to a
\emph{graph game }played between the pursuers and the evader.

In the current paper, inspired by Isler and Karnad's recent work
\cite{isler2008role}, we study the role of \emph{information }in
\emph{cops-and-robber} (CR)\ games, an important version of graph-based PE. By
\textquotedblleft information\textquotedblright\ we mean specifically the
players' \emph{location}. For example, we expect that when the cops know the
robber's location they can do better than when the robber is \textquotedblleft
invisible\textquotedblright. Our goal is to make \ precise the term
\textquotedblleft better\textquotedblright.

Reviews of \ the graph theoretic CR\ literature appear in
\cite{alspach2006searching,NowBook,fomin2008annotated}. In the
\textquotedblleft classical\textquotedblright\ CR\ variant
\cite{nowakowski1983vertex} it is assumed that the cops always know the
robber's location and vice versa. The \textquotedblleft
invisible\textquotedblright\ variant, in which the cops cannot see the robber
(but the robber always sees the cops) has received less attention in the graph
theoretic literature; among the few papers which treat this case we mention
\cite{isler2007randomized,isler2008role,kehagias2013cops,dereniowski2013zero}
and also \cite{adler2002randomized} in which \emph{both }cops and robber are invisible.

Both the visible and invisible CR\ variants are natural models for discretized
robotic PE\ problems; the connection has been noted and exploited relatively
recently \cite{isler2007randomized,isler2008role,vieira2009scalable2}. If it
is further assumed that the robber is not actively trying to avoid capture
(the case of \emph{drunk} robber) we obtain a \emph{one-player} graph game;
this model has been used quite often in mobile robotics
\cite{gerkey2005parallel,hollinger2009efficient,hollinger2010improving,lau2006probabilistic,sarmiento2003efficient}
and especially (when assuming random robber movement) in publications such as
\cite{hsu2008point,kurniawati2008sarsop,pineau2007POMDP,smith2004heuristic,spaan2005perseus}, 
which utilize \emph{partially observable Markov decision processes }(POMDP,
\cite{hauskrecht2000value,littman1996efficient,monahan1982survey}). For a more
general overview of pursuit/evasion and search problems in robotics, the
reader is referred to \cite{chungsearch}; some of the works cited in this
paper provide a useful background to the current paper.

This paper is structured as follows. In Section \ref{sec02} we present
preliminary material, notation and the definition of the \textquotedblleft
classical\textquotedblright\ CR game; we also introduce several node and edge
CR \emph{variants}. In Section \ref{sec03} we define rigorously the \emph{cop
number} and \emph{capture time} for the classical CR game and the previously
introduced CR\ variants. In Section \ref{sec04} we study the \emph{cost of
visibility} (\emph{COV}). In Section \ref{sec05} we present algorithms which
compute capture time and optimal strategies for several CR\ variants. In
Section \ref{sec06} we further study COV using computational experiments.
Finally, in Section \ref{sec07} we summarize and present our conclusions.

\section{Preliminaries\label{sec02}}

\subsection{Notation}

\begin{enumerate}
\item We use the following notations for sets:$\ \mathbb{N}$ denotes $\left\{
1,2,\ldots\right\}  $; $\mathbb{N}_{0}$ denotes $\left\{  0,1,2,\ldots
\right\}  $; $\left[  K\right]  $ denotes $\left\{  1,\ldots,K\right\}  $;
$A-B=\left\{  x:x\in A,x\notin B\right\}  $; $\left\vert A\right\vert $
denotes the \emph{cardinality} of $A$ (i.e., the number of its elements).

\item A \emph{graph} $G=(V,E)$ consists of a \emph{node set} $V$ and an
\emph{edge set} $E$, where every $e\in E$ has the form $e=\left\{
x,y\right\}  \subseteq V$. In other words, we are concerned with finite,
undirected, simple graphs; in addition we will always assume that $G$ is
connected and that $G$ contains $n$ nodes:\ $\left\vert V\right\vert =n$.
Furthermore, we will assume, without loss of generality, that the node set is
$V=\left\{  1,2,...,n\right\}  $. We let $V^{K}=\underset{K\text{ times}%
}{\underbrace{V\times V\times\ldots\times V}}$. We also define $V_{D}%
^{2}\subseteq V^{2}$ by $V_{D}^{2}=\{(x,x):x\in V\}$ (it is the set of
\textquotedblleft diagonal\textquotedblright\ node pairs).

\item A \emph{directed graph} (\emph{digraph}) $G=(V,E)$ consists of a
\emph{node set} $V$ and an \emph{edge set} $E$, where every $e\in E$ has the
form $e=\left(  x,y\right)  \in V\times V$. In other words, the edges of a
digraph are \emph{ordered} pairs.

\item In graphs, the (\emph{open}) \emph{neighborhood} of some $x\in V$ is
$N\left(  x\right)  =\left\{  y:\left\{  x,y\right\}  \in E\right\}  $; in
digraphs it is $N\left(  x\right)  =\left\{  y:\left(  x,y\right)  \in
E\right\}  $. In both cases, the \emph{closed neighborhood} of $x\ $is
$N\left[  x\right]  =N\left(  x\right)  \cup\left\{  x\right\}  $.

\item Given a graph $G=\left(  V,E\right)  $, its \emph{line graph} $L\left(
G\right)  =\left(  V^{\prime},E^{\prime}\right)  $ is defined as follows: the
node set is $V^{\prime}=E$, i.e., it has one node for every edge of $G$; the
edge set is defined by having the nodes $\left\{  u,v\right\}  ,\left\{
x,y\right\}  \in V^{\prime}$ connected by an edge $\left\{  \left\{
u,v\right\}  ,\left\{  x,y\right\}  \right\}  $ if and only if $\left\vert
\left\{  u,v\right\}  \cap\left\{  x,y\right\}  \right\vert =1$ (i.e., if the
original edges of $G$ are adjacent).

\item We will write $f\left(  n\right)  =o\left(  g\left(  n\right)  \right)
$ if and only if $\lim_{n\rightarrow\infty}\frac{f\left(  n\right)  }{g\left(
n\right)  }=0$. Note that in this \emph{asymptotic }notation $n$ denotes the
parameter with respect to which asymptotics are considered. So in later
sections we will write $o\left(  n\right)  $, $o\left(  M\right)  $ etc.
\end{enumerate}

\subsection{The CR\ Game Family}

The \textquotedblleft classical\textquotedblright\ CR\ game can be described
as follows. Player C controls $K$ cops (with $K\geq1$) and player R\ controls
a single robber. Cops and robber are moved along the edges of a graph
$G=\left(  V,E\right)  $ in discrete time steps $t\in\mathbb{N}_{0}$.\ \ At
time $t$, the robber's location is $Y_{t}\in V$ and the cops' locations are
$X_{t}=(X_{t}^{1},X_{t}^{2},\ldots,X_{t}^{K})\in V^{K}$ (for $t\in
\mathbb{N}_{0}$ and $k\in\left[  K\right]  $). The game is played in
\emph{turns}; in the 0-th turn first C places the cops on nodes of the graph
and then R places the robber; in the $t$-th turn, for $t>0$, \emph{first} C
moves the cops to $X_{t}$ and \emph{then} R moves the robber to $Y_{t}$. Two
types of moves are allowed:\ (a)\ sliding along a single edge and (b)\ staying
in place; in other words, for all $t$ and $k$, either $\{X_{t-1}^{k},X_{t}%
^{k}\}\in E$ or $X_{t-1}^{k}=X_{t}^{k}$; similarly, $\{Y_{t-1},Y_{t}\}\in E$
or $Y_{t-1}=Y_{t}$. The cops win if they \emph{capture} the robber, i.e., if
there exist $t\in\mathbb{N}_{0}$ and $k\in\left[  K\right]  $ such that
$Y_{t}=X_{t}^{k}$; the robber wins if for all $t\in\mathbb{N}_{0}$ and
$k\in\left[  K\right]  $ we have $Y_{t}\neq X_{t}^{k}$. In what follows we
will describe these eventualities by the following \textquotedblleft shorthand
notation\textquotedblright: $Y_{t}\in X_{t}$ and $Y_{t}\notin X_{t}$ (i.e., in
this notation we consider $X_{t}$ as a \emph{set} of cop positions).

In the classical game both C and R\ are \emph{adversarial}:\ C plays to effect
capture and R plays to avoid it. But there also exist \textquotedblleft%
\emph{drunk robber}\textquotedblright\ versions, in which the robber simply
performs a \emph{random walk }on $G$ such that, for all $\forall u,v\in V$ we
have%
\begin{equation}
\Pr\left(  Y_{0}=u\right)  =\frac{1}{n}\qquad\text{ and }\qquad\Pr\left(
Y_{t+1}=u|Y_{t}=v\right)  =\left\{
\begin{array}
[c]{ll}%
\frac{1}{\left\vert N\left(  v\right)  \right\vert } & \text{if and only if
}u\in N(v)\\
0 & \text{otherwise}%
\end{array}
\right.  .\label{eq022a}%
\end{equation}
In this case we can say that no R player is present (or, following a common
formulation, we can say that the R player is \textquotedblleft
Nature\textquotedblright).

If an R player exists, the cops' locations are always known to him; on the
other hand, the robber can be either \emph{visible} (his location is known to
C) or \emph{invisible }(his location is unknown). Hence we have four different
CR\ variants, as detailed in Table \ref{table01}.

\begin{table}[H]
\begin{center}%
\begin{tabular}
[c]{|l|l|}\hline
Adversarial Visible Robber & av-CR\\\hline
Adversarial Invisible Robber & ai-CR\\\hline
Drunk Visible Robber & dv-CR\\\hline
Drunk Invisible Robber & di-CR\\\hline
\end{tabular}
\end{center}
\caption{Four variants of the CR\ game.}%
\label{table01}%
\end{table}

In all of the above CR\ variants both cops and robber move from node to node.
This is a good model for entities (e.g., robots)\ which move from room to room
in an indoor environment. There also exist cases (for example moving in a maze
or a road network)\ where it makes more sense to assume that both cops and
robber move from edge to edge. We will call the classical version of the edge
CR game \emph{edge av-CR};\ it has attracted attention only recently
\cite{dudek2013cops}. Edge ai-CR, dv-CR and di-CR variants are also possible,
in analogy to the node versions listed in Table \ref{table01}. Each of these
cases can be reduced to the corresponding node variant, with the edge game
taking place on the line graph $L\left(  G\right)  $ of $G$.

\section{Cop Number and Capture Time\label{sec03}}

Two graph parameters which can be obtained from the av-CR game are the
\emph{cop number} and the \emph{capture time}. In this section we will define
these quantities in game theoretic terms 
\footnote{While this approach is not common in the
CR\ literature, we believe it offers certain advantages in clarity of
presentation.}
and also consider their extensions to
other CR variants. Before examining each of these CR variants in detail, let us
mention a particular modification which we will apply to all of them. Namely,
we assume that (every variant of) the CR game is played for an \emph{infinite
number of rounds}. This is obviously the case if the robber is never captured;
but we also assume that, in case the robber is captured at some time $t^{\ast
}$, the game continues for $t\in\{t^{\ast}+1,t^{\ast}+2,\ldots\}$ with the
following restriction:\ for all $t\geq t^{\ast}$, we have $Y_{t}%
=X_{t}^{k^{\ast}}$ (where $k^{\ast}$ is the number of cop who effected the
capture)\footnote{This modification facilitates the game theoretic analysis
presented in the sequel. Intuitively, it implies that after capture, the
$k^{*}$-th cop forces the robber to \textquotedblleft follow\textquotedblright%
\ him.}.

\subsection{\noindent\textbf{The Node av-CR Game\label{sec0301}}}

We will define cop number and capture time in game theoretic terms. To this
end we must first define \emph{histories} and \emph{strategies}.

A particular instance of the CR\ game can be fully described by the sequence
of cops and robber locations; these locations are fully determined by the C
and R\ moves. So, if we let $x_{t}\in V^{K}$ (resp. $y_{t}\in V$) denote the
nodes into which C (resp. R)\ places the cops (resp. the robber)\ at time $t$,
then a \emph{history} is a sequence $x_{0}y_{0}x_{1}y_{1}\ldots$ . Such a
sequence can have finite or infinite length; we denote the set of all finite
length histories by $H_{\ast}^{\left(  K\right)  }$; note that there exists an
infinite number of finite length sequences. By convention $H_{\ast}^{\left(
K\right)  }$ also includes the \emph{zero-length} or \emph{null} history,
which is the \emph{empty sequence}\footnote{This corresponds to the beginning
of the game, when neither player has made a move, just before C\ places the
cops on $G$.}, denoted by $\lambda$.\ Finally, we denote the set of all
infinite length histories by $H_{\infty}^{\left(  K\right)  }$.

Since both cops and robber are visible and the players move sequentially,
av-CR is a game of \emph{perfect information};\emph{ }in such a game C loses
nothing by limiting himself to \emph{pure }(i.e.,
deterministic)\ \emph{strategies }\cite{kuhn1950extensive}. A pure cop
strategy is a function $s_{C}:H_{\ast}^{\left(  K\right)  }\rightarrow V^{K}$;
a pure robber strategy is a function $s_{R}:H_{\ast}^{\left(  K\right)
}\rightarrow V$. In both cases the idea is that, given a finite length
history, the strategy produces the next cop or robber move\footnote{Note the
dependence on $K$, the number of cops.}; for example, when the robber strategy
$s_{R}$ receives the input $x_{0}$, it will produce the output $y_{0}%
=s_{R}\left(  x_{0}\right)  $; when it receives $x_{0}y_{0}x_{1}$, it will
produce $y_{1}=s_{R}\left(  x_{0}y_{0}x_{1}\right)  $ and so on. We will
denote the set of all legal cop strategies by $\mathbf{S}_{C}^{\left(
K\right)  }$ and the set of all legal robber strategies by $\mathbf{S}_{R}^{\left(
K\right)  }$; a strategy is \textquotedblleft legal\textquotedblright\ if it
only provides moves which respect the CR\ game rules. The set $\widetilde
{\mathbf{S}}_{C}^{\left(  K\right)  }\subseteq\mathbf{S}_{C}^{\left(
K\right)  }$ (resp. $\widetilde{\mathbf{S}}_{R}^{\left(  K\right)  }%
\subseteq\mathbf{S}_{R}^{\left(  K\right)  }$ ) is the set of
\emph{memoryless} legal cop (resp. robber)\ strategies, i.e., strategies which
only depend only on the current cops and robber positions; we will denote the
memoryless strategies by Greek letters, e.g., $\sigma_{C}$, $\sigma_{R}$ etc.
In other words%
\begin{align*}
\sigma_{C} &  \in\widetilde{\mathbf{S}}_{C}^{\left(  K\right)  }%
\Rightarrow\left[  \forall t:x_{t+1}=\sigma_{C}\left(  x_{0}y_{0}\ldots
.x_{t}y_{t}\right)  =\sigma_{C}\left(  x_{t}y_{t}\right)  \right]  ,\\
\sigma_{R} &  \in\widetilde{\mathbf{S}}_{R}^{\left(  K\right)  }%
\Rightarrow\left[  \forall t:y_{t+1}=\sigma_{R}\left(  x_{0}y_{0}\ldots
.x_{t}y_{t}x_{t+1}\right)  =\sigma_{R}\left(  y_{t}x_{t+1}\right)  \right]  .
\end{align*}

It seems intuitively obvious that both C\ and R\ lose nothing by playing with
memoryless strategies (i.e., computing their next moves based on the current
position of the game, not on its entire history). This is true but requires a
proof. One approach to this proof is furnished in
\cite{bonatogeneral,hahn2006note}. But we will present another proof by
recognizing that the CR\ game belongs to the extensively researched family of
\emph{reachability games} \cite{berwanger,mazala2002infinite}.

A reachability game is played by two players (Player 0 and Player 1)\ on a
\emph{digraph} $\overline{G}=\left(  \overline{V},\overline{E}\right)  $; each
node $v\in\overline{V}$ is a \emph{position} and each edge is a \emph{move};
i.e., the game moves from node to node (position) along the edges of the
digraph. The game is described by the tuple $\left(  \overline{V}%
_{0},\overline{V}_{1},\overline{E},\overline{F}\right)  $, where $\overline
{V}_{0}\cup\overline{V}_{1}=\overline{V}$, $\overline{V}_{0}\cap\overline
{V}_{1}=\emptyset$ and $\overline{F}\subseteq\overline{V}$. For $i\in\left\{
0,1\right\}  $, $\overline{V}_{i}$ is the set of positions (nodes)\ from which
the $i$-th Player makes the next move; the game terminates with a win for
Player 0 if and only if a move takes place into a node $v\in\overline{F}$ (the
\emph{target set} of Player 0); if this never happens, Player 1
wins\footnote{Here is a more intuitive description of the game: each move
consists in sliding a \emph{token} from one digraph node to another, along an
edge; the $i$-th player slides the token if and only if it is currently
located on a node $v\in\overline{V}_{i}$ ($i\in\left\{  0,1\right\}  $);
Player 0 wins if and only if the token goes into a node $u\in\overline{F}$;
otherwise Player 1 wins. }. The following is well known
\cite{berwanger,mazala2002infinite}.

\begin{theorem}
\label{prp0301a}Let $\left(  \overline{V}_{0},\overline{V}_{1},\overline
{E},\overline{F}\right)  $ be a reachability game on the digraph $\overline{D}=\left(
\overline{V},\overline{E}\right)  $. Then $\overline{V}$ can be partitioned
into two sets $\overline{W}_{0}$ and $\overline{W}_{1}$ such that (for
$i\in\left\{  0,1\right\}  $)\ player $i$ has a \emph{memoryless} strategy
$\sigma_{i}$ which is winning whenever the game starts in $u\in\overline
{W}_{i}$.
\end{theorem}

We can convert the av-CR game with $K$ cops to an equivalent reachability
game which is played on the \emph{CR\ game digraph}. In this digraph every
node corresponds to a \emph{position} of the original CR\ game; a (directed)
edge from node $u$ to node $v$ indicates that it is possible to get from
position $u$ to position $v$ in a single move. The CR\ game digraph has three
types of nodes.

\begin{enumerate}
\item Nodes of the form $u=\left(  x,y,p\right)  $ correspond to positions (in
the original CR\ game) with the cops located at $x\in V^{K}$, the robber at
$y\in V$ and player $p\in\left\{  C,R\right\}  $ being next to move.

\item There is single node $u=\left(  \lambda,\lambda,C\right)  $ which
corresponds to the starting position of the game: neither the cops nor the
robber have been placed on $G$; it is C's turn to move (recall that $\lambda$
denotes the empty sequence).

\item Finally, there exist $n$ nodes of the form $u=\left(  x,\lambda
,R\right)  $: the cops have just been placed in the graph (at positions $x\in
V^{K}$) but the robber has not been placed yet; it is R's turn to move.
\end{enumerate}

\noindent Let us now define%
\begin{align*}
\overline{V}_{0}^{\left(  K\right)  } &  =\left\{  \left(  x,y,C\right)  :x\in
V^{K}\cup\left\{  \lambda\right\}  ,y\in V\cup\left\{  \lambda\right\}
\right\}  ,\\
\overline{V}_{1}^{\left(  K\right)  } &  =\left\{  \left(  x,y,R\right)  :x\in
V^{K}\cup\left\{  \lambda\right\}  ,y\in V\cup\left\{  \lambda\right\}
\right\}  \\
\overline{V}^{\left(  K\right)  } &  =\overline{V}_{0}^{\left(  K\right)
}\cup\overline{V}_{1}^{\left(  K\right)  }%
\end{align*}
and let $\overline{E}^{\left(  K\right)  }$ consist of all pairs $\left(
u,v\right)  $ where $u,v\in\overline{V}^{\left(  K\right)  }$ and the move
from $u$ to $v$ is legal. Finally, we recognize that C's target set is
\[
\overline{F}^{\left(  K\right)  }=\left\{  \left(  x,y,p\right)  :x\in
V^{K},y\in\left(  V\cap x\right)  ,p\in\left\{  C,R\right\}  \right\}  .
\]
i.e., the set of all positions in which the robber is in the same node as at
least one cop.

With the above definitions, we have mapped the classical CR\ game (played with
$K$ cops on the graph $G$) to the reachability game $\left(  \overline{V}%
_{0}^{\left(  K\right)  },\overline{V}_{1}^{\left(  K\right)  },\overline
{E}^{\left(  K\right)  },\overline{F}^{\left(  K\right)  }\right)  $. By
Theorem \ref{prp0301a}, Player $i$ (with $i\in\left\{  0,1\right\}  $)\ will
have a \emph{winning set} $\overline{W}_{i}^{\left(  K\right)  }%
\subseteq\overline{V}^{\left(  K\right)  }$, i.e., a set with the following
property:\ whenever the reachability game starts at some $u\in\overline{W}%
_{i}^{\left(  K\right)  }$, then Player $i$ has a winning strategy (it may be
the case, for specific  $G$ and $K$ that either of $\overline{W}%
_{0}^{\left(  K\right)  }$, $\overline{W}_{1}^{\left(  K\right)  }$ is empty).
Recall that in our formulation of CR\ as a reachability  game, Player 0 is C.
In reachability terms, the statement \textquotedblleft C\ has a winning
strategy in the classical CR\ game\textquotedblright\ translates to
\textquotedblleft$\left(  \lambda,\lambda,C\right)  \in$ $\overline{W}%
_{0}^{\left(  K\right)  }$\textquotedblright\ and, for a given graph $G$, the
validity of this statement will in general depend on $K$ . It is clear that
$\overline{W}_{0}^{\left(  K\right)  }$ is increasing with $K$:
\begin{equation}
K_{1}\leq K_{2}\Rightarrow\overline{W}_{0}^{\left(  K_{1}\right)  }%
\subseteq\overline{W}_{0}^{\left(  K_{2}\right)  }.\label{eq0341}%
\end{equation}
It is also also clear that
\begin{equation}
\text{\textquotedblleft}\left(  \lambda,\lambda,C\right)  \in\overline{W}%
_{0}^{\left(  \left\vert V\right\vert \right)  }\text{\textquotedblright\ is
true for every }G=\left(  V,E\right)  \text{,}\label{eq0342}%
\end{equation}
because, if C has $\left\vert V\right\vert $ cops, he can place one in every
$u\in V$ and win immediately\footnote{In fact, for $K=\left\vert V\right\vert
$, we have $\overline{W}_{0}^{\left(  \left\vert V\right\vert \right)
}=\overline{V}^{\left(  K\right)  }$, because from every position $\left(
x,y,p\right)  $, C\ can move the cops so that one cop resides in each $u\in
V$, which guarantees immediate capture.}. 

Based on (\ref{eq0341})\ and (\ref{eq0342}) we can define the \emph{cop
number} of $G$ to be the minimum number of cops that guarantee capture; more
precisely we have the following definition (which is equivalent to the
\textquotedblleft classical\textquotedblright\ definition of cop number
\cite{aigner1984game}).

\begin{definition}
\label{prp0302a}The \emph{cop number} of $G$ is
\[
c\left(  G\right)  =\min\left\{  K:\left(  \lambda,\lambda,C\right)
\in\overline{W}_{0}^{\left(  K\right)  }\right\}  .
\]

\end{definition}

While a cop winning strategy $s_{C}$ guarantees that the token will go into
(and remain in)\ $\overline{F}^{\left(  K\right)  }$, we still do not know
how long it will take for this to happen. However, it is easy to prove that,
if $K\geq c(G)$ and C uses a \emph{memoryless} winning strategy, then no game
position will be repeated until capture takes place. Hence the following holds.

\begin{theorem}
\label{prp0302b}For every $G$, let $K\geq c\left(  G\right)  $ and consider
the CR\ game played on $G$ with $K$ cops. There exists a a memoryless cop
winning strategy $\sigma_{C}$ and a number $\overline{T}\left(  K;G\right)
<\infty$ such that, for every robber strategy $s_{R}$, C\ wins in no more than
$\overline{T}\left(  K;G\right)  $ rounds.
\end{theorem}

Let us now turn from winning to \emph{time optimal} strategies. To define
these, we first  define the \emph{capture time}, which will serve as the CR
\emph{payoff function}. 

\begin{definition}
\label{prp0302d}Given a graph $G$, some $K\in\mathbb{N}$ and strategies
$s_{C}\in\mathbf{S}_{C}^{\left(  K\right)  }$, $s_{R}\in\mathbf{S}%
_{R}^{\left(  K\right)  }$ the \emph{av-CR capture time} is defined by%
\begin{equation}
T^{\left(  K\right)  }\left(  s_{C},s_{R}|G\right)  =\min\left\{  t:\exists
k\in\left[  K\right]  \text{ such that }Y_{t}=X_{t}^{k}\right\}  ;
\label{eq0301}%
\end{equation}
in case capture never takes place, we let $T^{\left(  K\right)  }\left(
s_{C},s_{R}|G\right)  =\infty$.
\end{definition}

We will assume that R's \emph{payoff }is $T^{\left(  K\right)  }\left(
s_{C},s_{R}|G\right)  $ and C's payoff is $-T^{\left(  K\right)  }\left(
s_{C},s_{R}|G\right)  $ (hence av-CR is a \emph{two-person zero-sum game}).
Note that capture time (i) obviously\ depends on $K$ and (ii) for a fixed $K$
is fully determined by the $s_{C}$ and $s_{R}$ strategies. Now, following
standard game theoretic practice, we define \emph{optimal} strategies.

\begin{definition}
\label{prp0303a}For every graph $G$ and $K\in\mathbb{N}$, the strategies
$s_{C}^{\left(  K\right)  }\in\mathbf{S}_{C}^{\left(  K\right)  }$ and
$s_{R}^{\left(  K\right)  }\in\mathbf{S}_{R}^{\left(  K\right)  }$ are a pair
of \emph{optimal strategies} if and only if
\begin{equation}
\sup_{s_{R}\in\mathbf{S}_{R}^{\left(  K\right)  }}\inf_{s_{C}\in\mathbf{S}%
_{C}^{\left(  K\right)  }}T^{\left(  K\right)  }\left(  s_{C},s_{R}|G\right)
=\inf_{s_{C}\in\mathbf{S}_{C}^{\left(  K\right)  }}\sup_{s_{R}\in
\mathbf{S}_{R}^{\left(  K\right)  }}T^{\left(  K\right)  }\left(  s_{C}%
,s_{R}|G\right)  .\label{eq0302}%
\end{equation}
The \emph{value }of the av-CR game played with $K$ cops is the common value
of the two sides of (\ref{eq0302}) and we denote it $T^{\left(  K\right)
}\left(  s_{C}^{\left(  K\right)  },s_{R}^{\left(  K\right)  }|G\right)  $.
\end{definition}

We emphasize that the validity of  (\ref{eq0302}) is not known \emph{a
priori}. C (resp. R)\ can guarantee that he loses no more than $\inf_{s_{C}%
\in\mathbf{S}_{C}^{\left(  K\right)  }}\sup_{s_{R}\in\mathbf{S}_{R}^{\left(
K\right)  }}T^{\left(  K\right)  }\left(  s_{C},s_{R}|G\right)  $ (resp. gains
no less than $\sup_{s_{R}\in\mathbf{S}_{R}^{\left(  K\right)  }}\inf_{s_{C}%
\in\mathbf{S}_{C}^{\left(  K\right)  }}T^{\left(  K\right)  }\left(
s_{C},s_{R}|G\right)  $). We always have%
\begin{equation}
\sup_{s_{R}\in\mathbf{S}_{R}^{\left(  K\right)  }}\inf_{s_{C}\in\mathbf{S}%
_{C}^{\left(  K\right)  }}T^{\left(  K\right)  }\left(  s_{C},s_{R}|G\right)
\leq\inf_{s_{C}\in\mathbf{S}_{C}^{\left(  K\right)  }}\sup_{s_{R}\in
\mathbf{S}_{R}^{\left(  K\right)  }}T^{\left(  K\right)  }\left(  s_{C}%
,s_{R}|G\right)  .\label{eq0303}%
\end{equation}
But, since av-CR is an \emph{infinite} game (i.e., depending on $s_{C}$ and $s_{R}$, it can last an
infinite number of turns) it is not clear that equality holds in
(\ref{eq0303}) and, even when it does, the existence of optimal strategies
$\left(  s_{C}^{\left(  K\right)  },s_{R}^{\left(  K\right)  }\right)  $ which
achieve the value is not guaranteed.

In fact it can be proved that, for $K\geq c\left(  G\right)  $, av-CR has both
a value and optimal strategies. The details of this proof will be reported
elsewhere, but the gist of the argument is the following. Since av-CR\ is
played with $K\geq c\left(  G\right)  $ cops, by Theorem \ref{prp0302b},
C\ has a \emph{memoryless} strategy which guarantees the game will last no
more than $\overline{T}\left(  K;G\right)  $ turns. Hence av-CR\ with $K\geq
c\left(  G\right)  $ essentially is a \emph{finite} zero-sum two-player game;
it is well known \cite{osborne1994course} that every such game has a value and
optimal memoryless strategies. In short, we have the following.

\begin{theorem}
\label{prp0308a}Given any graph $G$ and any $K\geq c\left(  G\right)  $, for
the av-CR game there exists a pair $\left(  \sigma_{C}^{\left(  K\right)
},\sigma_{R}^{\left(  K\right)  }\right)  \in\widetilde{\mathbf{S}}%
_{C}^{\left(  K\right)  }\times\widetilde{\mathbf{S}}_{R}^{\left(  K\right)
}$ of \emph{memoryless time optimal strategies} such that
\[
T^{\left(  K\right)  }\left(  \sigma_{C}^{\left(  K\right)  },\sigma
_{R}^{\left(  K\right)  }|G\right)  =\sup_{s_{R}\in\mathbf{S}_{R}^{\left(
K\right)  }}\inf_{s_{C}\in\mathbf{S}_{C}^{\left(  K\right)  }}T^{\left(
K\right)  }\left(  s_{C},s_{R}|G\right)  =\inf_{s_{C}\in\mathbf{S}%
_{C}^{\left(  K\right)  }}\sup_{s_{R}\in\mathbf{S}_{R}^{\left(  K\right)  }%
}T^{\left(  K\right)  }\left(  s_{C},s_{R}|G\right)  .
\]

\end{theorem}

Hence we can define the capture time of a graph to be the value of av-CR when
played on $G$ with $K=c\left(  G\right)  $ cops.

\begin{definition}
\label{prp0303}The \emph{adversarial visible capture time} of $G$ is
\[
ct\left(  G\right)  =\sup_{s_{R}\in\mathbf{S}_{R}^{\left(  K\right)  }}%
\inf_{s_{C}\in\mathbf{S}_{C}^{\left(  K\right)  }}T^{\left(  K\right)
}\left(  s_{C},s_{R}|G\right)  =\inf_{s_{C}\in\mathbf{S}_{C}^{\left(
K\right)  }}\sup_{s_{R}\in\mathbf{S}_{R}^{\left(  K\right)  }}T^{\left(
K\right)  }\left(  s_{C},s_{R}|G\right)  .
\]
with $K=c\left(  G\right)  $.
\end{definition}

\subsection{\noindent\textbf{The Node dv-CR Game\label{sec0302}}}

In this game the robber is visible and performs a random walk on $G$
(\emph{drunk} robber)\ as indicated by (\ref{eq022a}). In the absence of cops,
$Y_{t}$ is a Markov chain on $V$, with transition probability matrix $P$,
where for every $u,v\in\{1,2,...,|V|\}$ we have
\[
P_{u,v}=\Pr\left(  Y_{t+1}=u|Y_{t}=v\right)  .
\]
In the presence of one or more cops, $\left\{  Y_{t}\right\}  _{t=0}^{\infty}$
$\ $is a \emph{Markov decision process} (MDP) \cite{puterman1994markov} with
state space $V\cup\left\{  n+1\right\}  $ (where $n+1$ is the \emph{capture
state}) and transition probability matrix $P\left(  X_{t}\right)  $ (obtained
from $P$ as shown in \cite{kehagias2012some}); in other words, $X_{t}$ is the
\emph{control} variable, selected by C.

Since no robber strategy is involved, the capture time on $G$ only depends on
the ($K$-cops strategy) $s_{C}$: namely:
\begin{equation}
T^{\left(  K\right)  }\left(  s_{C}|G\right)  =\min\left\{  t:\exists
k\in\left[  K\right]  \text{ such that }Y_{t}=X_{t}^{k}\right\}
,\label{eq0332}%
\end{equation}
which can also be written as
\begin{equation}
T^{\left(  K\right)  }\left(  s_{C}|G\right)  =\sum_{t=0}^{\infty}%
\mathbf{1}\left(  Y_{t}\notin X_{t}\right)  ,\label{eq0333}%
\end{equation}
where $\mathbf{1}\left(  Y_{t}\notin X_{t}\right)  $ equals $1$ if $Y_{t}$
does not belong to $X_{t}$ (taken as a set of cop positions)\ and 0 otherwise.
Since the robber performs a random walk on $G$, it follows that $T^{\left(  K\right)  }\left(
s_{C}|G\right)  $ is a random variable, and C wants to minimize its expected
value:%
\begin{equation}
E\left(  T^{\left(  K\right)  }\left(  s_{C}|G\right)  \right)  =E\left(
\sum_{t=0}^{\infty}\mathbf{1}\left(  Y_{t}\notin X_{t}\right)  \right)
.\label{eq0331}%
\end{equation}
The minimization of (\ref{eq0331})\ is a typical \emph{undiscounted, infinite
horizon} MDP\ problem. Using standard MDP\ results \cite{puterman1994markov}
we see that (i) C loses nothing by determining $X_{0},X_{1},\ldots$ through a
memoryless strategy $\sigma_{C}\left(  x,y\right)  $ and (ii) for every
$K\geq1$, $E\left(  T^{\left(  K\right)  }\left(  \sigma_{C}|G\right)
\right)  $ is well defined. Furthermore, for every $K\in\mathbb{N}$ there
exists an optimal strategy $\sigma_{C}^{\left(  K\right)  }$ which minimizes
$E\left(  T^{\left(  K\right)  }\left(  \sigma_{C}|G\right)  \right)  $; hence
we have the following.

\begin{theorem}
\label{prp0310a}Given any graph $G$ and $K\in\mathbb{N}$, for the dv-CR\ game
played on $G$ with $K$ cops there exists a memoryless strategy $\sigma
_{C}^{\left(  K\right)  }\in\widetilde{\mathbf{S}}_{C}^{\left(  K\right)  }$
such that
\[
E\left(  T^{\left(  K\right)  }\left(  \sigma_{C}^{\left(  K\right)
}|G\right)  \right)  =\inf_{s_{C}\in\mathbf{S}_{C}^{\left(  K\right)  }%
}E\left(  T^{\left(  K\right)  }\left(  s_{C}|G\right)  \right)  .
\]

\end{theorem}

\begin{definition}
The \emph{drunk visible capture time} of $G$ is
\[
dct\left(  G\right)  =\inf_{s_{C}\in\mathbf{S}_{C}^{\left(  K\right)  }%
}E\left(  T^{\left(  K\right)  }\left(  s_{C}|G\right)  \right)  .
\]
with $K=c\left(  G\right)  $.
\end{definition}

Note that, even though a single cop suffices to capture the drunk robber on
any $G$, we have chosen to define $dct\left(  G\right)  $ to be the capture
time for $K=c\left(  G\right)  $ cops; we have done this to make (in Section
\ref{sec04}) an equitable comparison between $ct\left(  G\right)  $ and
$dct\left(  G\right)  $.

\subsection{\noindent\textbf{The Node ai-CR Game\label{sec0303}}}

This is \emph{not} a perfect information game, 
since C\ cannot see R's moves. Hence C and R\ must use \emph{mixed} strategies $s_{C}$, $s_{R}$. A mixed
strategy $s_{C}$ (resp. $s_{R}$) specifies, for every $t$, a conditional
probability $\Pr\left(  X_{t}|X_{0},Y_{0},\ldots,Y_{t-2},X_{t-1}%
,Y_{t-1}\right)  $ (resp. $\Pr\left(  Y_{t}|X_{0},Y_{0},\ldots,Y_{t-1}%
,X_{t}\right)  $) according to which C (resp. R)\ selects his $t$-th move. Let
$\overline{\mathbf{S}}_{C}^{\left(  K\right)  }$(resp. $\overline{\mathbf{S}%
}_{R}^{\left(  K\right)  }$) be the set of all mixed cop (resp.
robber)\ strategies. A strategy pair $\left(  s_{R},s_{C}\right)  \in
\overline{\mathbf{S}}_{C}^{\left(  K\right)  }\times\overline{\mathbf{S}}%
_{R}^{\left(  K\right)  }$, specifies probabilities for all events $\left(
X_{0}=x_{0},\ldots,X_{t}=x_{t},Y_{0}=y_{0},\ldots,Y_{t}=y_{t}\right)  $ and
these induce a probability measure which in turn determines R's expected gain
(and C's expected loss), namely $E\left(  T^{\left(  K\right)  }\left(
s_{C}^{\left(  K\right)  },s_{R}^{\left(  K\right)  }|G\right)  \right)  $.
Let us define%
\begin{align*}
\underline{v}^{\left(  K\right)  }  &  =\sup_{s_{R}\in\overline{\mathbf{S}%
}_{R}^{\left(  K\right)  }}\inf_{s_{C}\in\overline{\mathbf{S}}_{C}^{\left(
K\right)  }}E\left(  T^{\left(  K\right)  }\left(  s_{C},s_{R}|G\right)
\right)  ,\\
\overline{v}^{\left(  K\right)  }  &  =\inf_{s_{C}\in\overline{\mathbf{S}}%
_{C}^{\left(  K\right)  }}\sup_{s_{R}\in\overline{\mathbf{S}}_{R}^{\left(
K\right)  }}E\left(  T^{\left(  K\right)  }\left(  s_{C},s_{R}|G\right)
\right)  .
\end{align*}
Similarly to av-CR, C (resp. R)\ can guarantee an expected payoff no greater
than $\overline{v}^{\left(  K\right)  }$ (resp. no less than \underline{$v$%
}$^{\left(  K\right)  }$). If \underline{$v$}$^{\left(  K\right)  }%
=\overline{v}^{\left(  K\right)  }$, we denote the common value by $v^{\left(
K\right)  }$ and call it the \emph{value} of the ai-CR game (played on $G$,
with $K$ cops). A pair of strategies $\left(  s_{C}^{\left(  K\right)  }%
,s_{R}^{\left(  K\right)  }\right)  $ is called optimal if and only if
$E\left(  T^{\left(  K\right)  }\left(  s_{C}^{\left(  K\right)  }%
,s_{R}^{\left(  K\right)  }|G\right)  \right)  =v^{\left(  K\right)  }$.

In \cite{kehagias2013cops} we have studied the ai-CR game and proved that it
does indeed have a value and optimal strategies. We give a summary of the
relevant argument; proofs can be found in \cite{kehagias2013cops}.

First, \emph{invisibility does not increase the cop number}. In other words,
there is a cop strategy (involving $c\left(  G\right)  $ cops)\ which
guarantees bounded expected capture time for every robber strategy $s_{R}$.
More precisely, we have proved the following.

\begin{theorem}
\label{prp0307}On any graph $G$ let $\overline{s}_{C}^{\left(  K\right)  }$
denote the strategy in which $K$ cops random-walk on $G$. Then
\[
\forall K\geq c\left(  G\right)  :\sup_{s_{R}\in\overline{\mathbf{S}}%
_{R}^{\left(  K\right)  }}E\left(  T^{\left(  K\right)  }\left(  \overline
{s}_{C}^{\left(  K\right)  },s_{R}|G\right)  \right)  <\infty.
\]

\end{theorem}

Now consider the \textquotedblleft$m$-truncated\textquotedblright%
\ ai-CR$\ $game which is played exactly as the \textquotedblleft
regular\textquotedblright\ ai-CR$\ $but lasts at most $m$ turns. Strategies
$s_{R}\in\overline{\mathbf{S}}_{R}^{\left(  K\right)  }$ and $s_{C}%
\in\overline{\mathbf{S}}_{C}^{\left(  K\right)  }$ can be used in the
$m$-truncated game: C and R use them only until the $m$-th turn. Let R receive
one payoff unit for every turn in which the robber is not captured; denote the
payoff  of the $m$-truncated game (when strategies $s_{C}$,$\ s_{R}$ are used)
by $T_{m}^{\left(  K\right)  }\left(  s_{C},s_{R}|G\right)  $. Clearly
\[
\forall m\in\mathbb{N},s_{R}\in\overline{\mathbf{S}}_{R}^{\left(  K\right)
},s_{C}\in\overline{\mathbf{S}}_{C}^{\left(  K\right)  }:T_{m}^{\left(
K\right)  }\left(  s_{C},s_{R}|G\right)  \leq T_{m+1}^{\left(  K\right)
}\left(  s_{C},s_{R}|G\right)  \leq T^{\left(  K\right)  }\left(  s_{C}%
,s_{R}|G\right)
\]
The expected payoff of the $m$-truncated game is $E\left(  T_{m}^{\left(
K\right)  }\left(  s_{C},s_{R}|G\right)  \right)  $. Because it is a
\emph{finite}, two-person, zero-sum game, the $m$-truncated game has a value
and optimal strategies. Namely, the value is
\[
v^{\left(  K,m\right)  }=\sup_{s_{R}\in\overline{\mathbf{S}}_{R}^{\left(
K\right)  }}\inf_{s_{C}\in\overline{\mathbf{S}}_{C}^{\left(  K\right)  }%
}E\left(  T_{m}^{\left(  K\right)  }\left(  s_{C},s_{R}|G\right)  \right)
=\inf_{s_{C}\in\overline{\mathbf{S}}_{C}^{\left(  K\right)  }}\sup_{s_{R}%
\in\overline{\mathbf{S}}_{R}^{\left(  K\right)  }}E\left(  T_{m}^{\left(
K\right)  }\left(  s_{C},s_{R}|G\right)  \right)
\]
and there exist optimal strategies $s_{C}^{\left(  K,m\right)  }\in
\overline{\mathbf{S}}_{C}^{\left(  K\right)  }$, $s_{R}^{\left(  K,m\right)
}\in\overline{\mathbf{S}}_{R}^{\left(  K\right)  }$ such that
\begin{equation}
E\left(  T_{m}^{\left(  K\right)  }\left(  s_{C}^{\left(  K,m\right)  }%
,s_{R}^{\left(  K,m\right)  }|G\right)  \right)  =v^{\left(  K,m\right)
}<\infty.\label{eq3002}%
\end{equation}
In \cite{kehagias2013cops} we use the truncated games to prove that the
\textquotedblleft regular\textquotedblright\ ai-CR game has a value, an
optimal C strategy and $\varepsilon$-optimal R strategies. More precisely, we
prove the following.

\begin{theorem}
\label{prp0308}Given any graph $G$ and $K\geq c\left(  G\right)  $, the
ai-CR\ game played on $G$ with $K$ cops has a value $v^{\left(  K\right)  }$
which satisfies%
\[
\lim_{m\rightarrow\infty}v^{\left(  K,m\right)  }=\underline{v}^{\left(
K\right)  }=\overline{v}^{\left(  K\right)  }=v^{\left(  K\right)  }.
\]
Furthermore, there exists a strategy $s_{C}^{\left(  K\right)  }\in
\overline{\mathbf{S}}_{C}^{\left(  K\right)  }$ such that
\begin{equation}
\sup_{s_{R}\in\overline{\mathbf{S}}_{R}^{\left(  K\right)  }}E\left(
T^{\left(  K\right)  }\left(  s_{C}^{\left(  K\right)  },s_{R}\right)
\right)  =v^{\left(  K\right)  }; \label{eq131a}%
\end{equation}
and for every $\varepsilon>0$ there exists an $m_{\varepsilon}$ and a strategy
$s_{R}^{\left(  K,\varepsilon\right)  }$ such that
\begin{equation}
\forall m\geq m_{\varepsilon}:v^{\left(  K\right)  }-\varepsilon\leq
\sup_{s_{C}\in\overline{\mathbf{S}}_{C}^{\left(  K\right)  }}E\left(
T^{\left(  K\right)  }\left(  s_{C},s_{R}^{\left(  K,\varepsilon\right)
}\right)  |G\right)  \leq v^{\left(  K\right)  }. \label{eq131b}%
\end{equation}

\end{theorem}

Having established the existence of $v^{\left(  K\right)  }$ we have the following.

\begin{definition}
The \emph{adversarial invisible capture time} of $G$ is
\[
ct_{i}\left(  G\right)  =v^{\left(  K\right)  }=\sup_{s_{R}\in\overline
{\mathbf{S}}_{R}^{\left(  K\right)  }}\inf_{s_{C}\in\overline{\mathbf{S}}%
_{C}^{\left(  K\right)  }}E\left(  T^{\left(  K\right)  }\left(  s_{C}%
,s_{R}|G\right)  \right)  =\inf_{s_{C}\in\overline{\mathbf{S}}_{C}^{\left(
K\right)  }}\sup_{s_{R}\in\overline{\mathbf{S}}_{R}^{\left(  K\right)  }%
}E\left(  T^{\left(  K\right)  }\left(  s_{C},s_{R}|G\right)  \right)
\]
with $K=c\left(  G\right)  $.
\end{definition}

\subsection{\noindent\textbf{The Node di-CR Game\label{sec0304}}}

In this game $Y_{t}$ is unobservable and drunk; call this the
\textquotedblleft regular\textquotedblright\ di-CR game and also introduce the
$m$-truncated di-CR game. Both are one-player \ games or, equivalently,
$Y_{t}$ is a \emph{partially observable MDP\ }(POMDP)
\cite{puterman1994markov}. The target function is
\begin{equation}
E\left(  T^{\left(  K\right)  }\left(  s_{C}|G\right)  \right)  =E\left(
\sum_{t=0}^{\infty}\mathbf{1}\left(  Y_{t}\notin X_{t}\right)  \right)
,\label{eq0311a}%
\end{equation}
which is exactly the same as (\ref{eq0331}) but now $Y_{t}$ is
\emph{unobservable}. (\ref{eq0311a}) can be approximated by
\begin{equation}
E\left(  T_{m}^{\left(  K\right)  }\left(  s_{C}|G\right)  \right)  =E\left(
\sum_{t=0}^{m}\mathbf{1}\left(  Y_{t}\notin X_{t}\right)  \right)
.\label{eq0312}%
\end{equation}
The expected values in (\ref{eq0311a})-(\ref{eq0312})\ are well defined for
every $s_{C}$. C must select a strategy $s_{C}\in\overline{\mathbf{S}}%
_{C}^{\left(  K\right)  }$ which minimizes $E\left(  T^{\left(  K\right)
}\left(  s_{C}|G\right)  \right)  $. This is a typical \emph{infinite horizon,
undiscounted} POMDP\ problem \cite{puterman1994markov} for which the following holds.

\begin{theorem}
\label{prp0310b}Given any graph $G$ and $K\in\mathbb{N}$, for the di-CR\ game
played on $G$ with $K$ cops there exists a strategy $s_{C}^{\left(  K\right)
}\in\overline{\mathbf{S}}_{C}^{\left(  K\right)  }$ such that
\[
E\left(  T^{\left(  K\right)  }\left(  s_{C}^{\left(  K\right)  }|G\right)
\right)  =\inf_{s_{C}\in\overline{\mathbf{S}}_{C}^{\left(  K\right)  }%
}E\left(  T^{\left(  K\right)  }\left(  s_{C}|G\right)  \right)  .
\]

\end{theorem}

Hence we can introduce the following.

\begin{definition}
The \emph{drunk invisible capture time} of $G$ is
\[
dct_{i}\left(  G\right)  =\inf_{s_{C}\in\overline{\mathbf{S}}_{C}^{\left(
K\right)  }}E\left(  T^{\left(  K\right)  }\left(  s_{C}|G\right)  \right)  .
\]
with $K=c\left(  G\right)  $.
\end{definition}

\subsection{The Edge CR Games\label{sec0305}}

As already mentioned, every edge CR\ variant can be reduced to the
corresponding node variant played on $L\left(  G\right)  $, the line graph of
$G$. Hence all the results and definitions of Sections \ref{sec0301} -
\ref{sec0304} hold for the edge variants as well. In particular, we have an
edge cop number $\overline{c}\left(  G\right)  =c\left(  L\left(  G\right)
\right)  $ and capture times
\[
\overline{ct}\left(  G\right)  =ct\left(  L\left(  G\right)  \right)
,\quad\overline{dct}\left(  G\right)  =dct\left(  L\left(  G\right)  \right)
,\quad\overline{ct}_{i}\left(  G\right)  =ct_{i}\left(  L\left(  G\right)
\right)  ,\quad\overline{dct}_{i}\left(  G\right)  =dct_{i}\left(  L\left(
G\right)  \right)  .
\]
In general, all of these \textquotedblleft edge
CR\ parameters\textquotedblright\ will differ from the corresponding
\textquotedblleft node CR parameters\textquotedblright.

\section{The Cost of Visibility\label{sec04}}

\subsection{Cost of Visibility in the Node CR\ Games}

As already remarked, we expect that ai-CR is more difficult (from C's point of
view) than av-CR (the same holds for the drunk counterparts of this game). We
quantify this statement by introducing the \emph{cost of visibility}
(\emph{COV}).

\begin{definition}
\label{prp0401}For every $G$, the \emph{adversarial cost of visibility} \ is
$H_{a}(G)=\frac{ct_{i}(G)}{ct(G)}$ and the \emph{drunk cost of visibility} is
$H_{d}(G)=\frac{dct_{i}(G)}{dct(G)}$.
\end{definition}

Clearly, for every $G$ we have $H_{a}\left(  G\right)  \geq1$ and
$H_{d}\left(  G\right)  \geq1$ (i.e., it is at least as hard to capture an
invisible robber than a visible one). The following theorem shows that in fact
both $H_{a}\left(  G\right)  $ and $H_{d}\left(  G\right)  $ can become
arbitrarily large. In proving the corresponding theorems, we will need the
family of \emph{long star graphs} $S_{N,M}$. For specific values of $M$ and
$N$, $S_{N,M}$ consists of $N$ paths (we call these \emph{rays}) each having
$M$ nodes, joined at a central node, as shown in Fig.\ref{fig001b}.
\begin{figure}[H]
\centering\scalebox{0.3}{\includegraphics{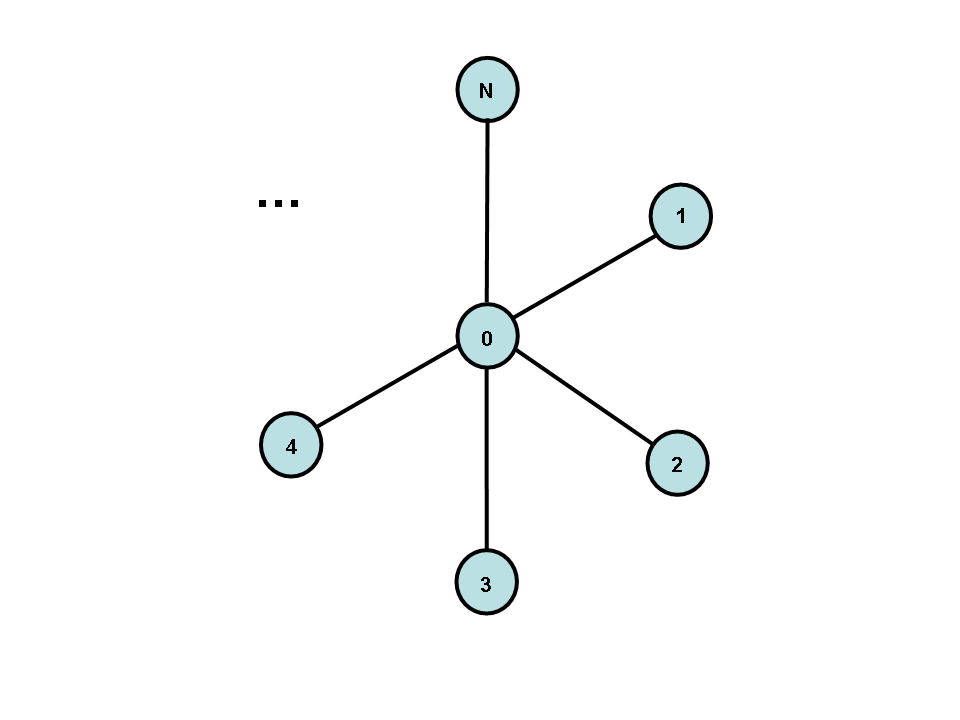}}\scalebox{0.3}{\includegraphics{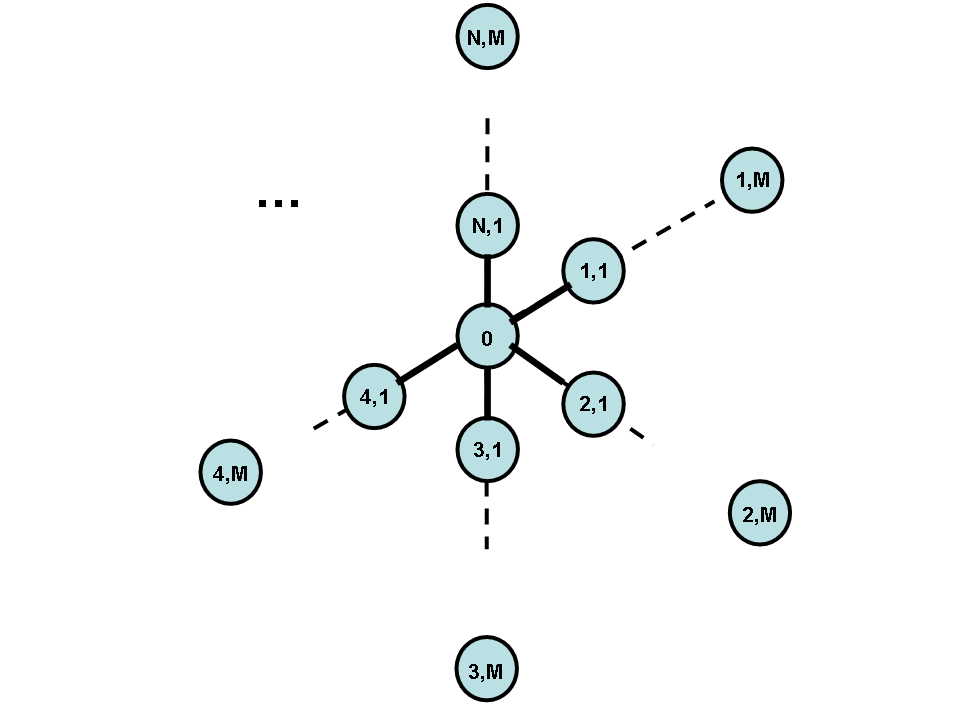}}\caption[The
broom graph $B_{n,c}$.]{(a) The \emph{star} graph $S_{N,1}$. (b) The
\emph{long star} graph $S_{N,M}$. }%
\label{fig001b}%
\end{figure}

\begin{theorem}
\label{prp0402}For every $N\in\mathbb{N}$ we have $H_{a}\left(  S_{N,1}%
\right)  =N$.
\end{theorem}

\begin{proof}
\noindent\underline{\textbf{(i) Computing }$ct\left(  S_{N,1}\right)  $}. In
av-CR, for every $N\in\mathbb{N}$ we have $ct\left(  S_{N,1}\right)  =1$: the
cop starts at $X_{0}=0$, the robber starts at some $Y_{0}=u\neq0$ and, at
$t=1$, he is captured by the cop moving into $u$; i.e., $ct\left(
S_{N,1}\right)  \leq1$; on the other hand, since there are at least two
vertices ($N\geq1$), clearly $ct\left(  S_{N,1}\right)  \geq1$.

\noindent\underline{\textbf{(ii) Computing }$ct_{i}\left(  S_{N,1}\right)  $}.
Let us now show that in ai-CR we have $ct_{i}\left(  S_{N,1}\right)  =N$. C
places the cop at $X_{0}=0$ and R places the robber at some $Y_{0}=u\neq0$. We
will obtain $ct_{i}\left(  S_{N,1}\right)  $ by bounding it from above and
below. For an upper bound, consider the following C strategy. Since C does not
know the robber's location, he must check the leaf nodes one by one. So at
every odd $t$ he moves the cop into some $u\in\left\{  1,2,\ldots,N\right\}  $
and at every even $t$ he returns to $0$. Note that R cannot change the
robber's original position; in order to do this, the robber must pass through
0 but then he will be captured by the cop (who either is already in 0 or will
be moved into it just after the robber's move). Hence C can choose the nodes
he will check on odd turns with uniform probability and \emph{without
repetitions}. Equivalently, we can assume that the order in which nodes are
chosen by C is selected uniformly at random from the set of all permutations;
further, we assume that R (who does not know this order) starts at some
$Y_{0}=u\in\left\{  1,\ldots,N\right\}  $. Then we have
\[
ct_{i}\left(  S_{N,1}\right)  \leq\frac{1}{N}\cdot1+\frac{1}{N}\cdot
3+\ldots+\frac{1}{N}\cdot\left(  2N-1\right)  =N.
\]
For a lower bound, consider the following R strategy. The robber is initially
placed at a random leaf that is different than the one selected by C (if the
cop did not start at the center). Knowing this, the best C strategy is to
check (in any order) all leaves without repetition. If the cop starts at the
center, we get exactly the same sum as for the upper bound. Otherwise, we
have
\[
ct_{i}\left(  S_{N,1}\right)  \geq\frac{1}{N-1}\cdot2+\frac{1}{N-1}%
\cdot4+\ldots+\frac{1}{N-1}\cdot\left(  2N-2\right)  =N.
\]

\noindent\underline{\textbf{(iii) Computing }$H_{a}\left(  S_{N,1}\right)  $}.
Hence, for all $N\in\mathbb{N}$ we have $H_{a}\left(  S_{N,1}\right)
=\frac{ct_{i}\left(  S_{N,1}\right)  }{ct\left(  S_{N,1}\right)  }=N$.
\end{proof}

\begin{theorem}
\label{prp0402b}For every $N\in\mathbb{N}-\{1\}$ we have
\[
H_{d}\left(  S_{N,M}\right)  =(1+o(1))\frac{(2N-1)(N-1)+1}{N}\geq2N-3,
\]
where the asymptotics is with respect to $M$; $N$ is considered a fixed constant.
\end{theorem}

\begin{proof}
\underline{\textbf{(i) Computing }$dct\left(  S_{N,M}\right)  $}. We will
first show that, for any $N\in\mathbb{N}$, we have $dct\left(  S_{N,M}\right)
=\left(  1+o\left(  1\right)  \right)  \frac{M}{2}$ (recall that the parameter
$N$ is a fixed constant whereas $M\rightarrow\infty$.) Suppose that the cop
starts on the $i$-th ray, at distance $(1+o(1))cM$ from the center (for some
constant $c\in\lbrack0,1]$). The robber starts at a random vertex. It follows
that for any $j$ such that $1\leq j\leq N$, the robber starts on the $j$-th
ray with probability $(1+o(1))/N$. It is a straightforward application of
Chernoff bounds\footnote{If $X$ has a binomial distribution $Bin(n,p)$, then
$Pr(|X-np|\geq\epsilon np)\leq2\exp(-\epsilon^{2}np/3)$ for any $\epsilon
\leq3/2$. Suppose the robber starts at distance $\omega(M^{2/3})$ from the
center. During $N=O(M)$ steps the robber makes in expectation $N/2$ steps
towards the center, and $N/2$ steps towards the end of the ray. The
probability to make during $N$ steps more than $N/2+M^{2/3}$ steps towards the
center, say, is thus at most $e^{-cM^{1/3}}$, and the same holds also by
taking a union bound over all $O(M)$ steps. Hence, with probability at least
$1-e^{-cM^{1/3}}$ he will throughout $O(M)$ steps remain at distance
$O(M^{2/3})$ from his initial position.} to show that with probability
$1+o(1)$ the robber will not move by more than $o(M)$ in the next $O(MN)=O(M)$
steps, which suffice to finish the game. Hence, the expected capture time is
easy to calculate.

\begin{itemize}
\item With probability $(1-c+o(1))/N$, the robber starts on the same ray as
the cop but farther away from the center. Conditioning on this event, the
expected capture time is $M(1-c+o(1))/2$.

\item With probability $(c+o(1))/N$, the robber starts on the same ray as the
cop but closer to the center. Conditioning on this event, the expected capture
time is $M(c+o(1))/2$.

\item With probability $(N-1+o(1))/N$, the robber starts on different ray than
the cop. Conditioning on this event, the expected capture time is $(c+o(1))M +
M(1/2+o(1))$.
\end{itemize}

\noindent It follows that the expected capture time is
\[
(1+o(1))M\left(  \frac{1-c}{N}\cdot\frac{1-c}{2}+\frac{c}{N}\cdot\frac{c}%
{2}+\frac{N-1}{N}\cdot\frac{2c+1}{2}\right)  ,
\]
which is maximized for $c=0$, giving $dct\left(  S_{N,M}\right)  =\left(
1+o\left(  1\right)  \right)  \frac{M}{2}$.

\noindent\underline{\textbf{(ii) Computing}\emph{ }$dct_{i}\left(
S_{N,M}\right)  $}. The initial placement for the robber is the same as in the
visible variant, that is, the uniform distribution is used. However, since the
robber is now invisible, C has to check all rays. As before, by Chernoff
bounds, with probability at least $1-e^{-cM^{1/3}}$ (for some constant $c>0$)
during $O(M)$ steps the robber is always within distance $O(M^{2/3})$ from its
initial position. If the robber starts at distance $\omega(M^{2/3})$ from the
center, he will thus with probability at least $1-e^{-cM^{1/3}}$ not change
his ray during $O(M)$ steps. Otherwise, he might change from one ray to the
other with bigger probability, but note that this happens only with the
probability of the robber starting at distance $O(M^{2/3})$ from the center,
and thus with probability at most $O(M^{-1/3})$. Keeping these remarks in
mind, let us examine \textquotedblleft reasonable\textquotedblright%
\ C\ strategies. It turns out there exist three such.

\noindent\textbf{(ii.1) }Suppose C starts at the end of one ray (chosen
arbitrarily), goes to the center, and then successively checks the remaining
rays without repetition, with probability at least $1-O(M^{-1/3})$, the robber
will be caught. If the robber does not switch rays (and is therefore caught),
the capture time is calculated as follows:

\begin{itemize}
\item With probability $(1+o(1))/N$, the robber starts on the same ray as the
cop. Conditioning on this event, the expected capture time is $(1+o(1)) M/2$.

\item With probability $(1+o(1))/N$, the robber starts on the $j$-th ray
visited by the cop. Conditioning on this event, the expected capture time is
$(1+o(1)) (M+2M(j-2)+M/2)$. ($M$ steps are required to move from the end of
the first ray to the center, $2M$ steps are `wasted' to check $j-2$ rays, and
$M/2$ steps are needed to catch the robber on the $j$-th ray, on expectation.)
\end{itemize}

\noindent\noindent Hence, conditioned under not switching rays, the expected
capture time in this case is
\begin{align*}
&  (1+o(1))\frac{M}{N}\left(  \frac{1}{2}+\left(  1+\frac{1}{2}\right)
+\left(  3+\frac{1}{2}\right)  +\ldots+\left(  1+2(N-2)+\frac{1}{2}\right)
\right)  \\
&  =(1+o(1))\frac{M}{N}\left(  \frac{1}{2}+\left(  2\cdot1-\frac{1}{2}\right)
+\left(  2\cdot2-\frac{1}{2}\right)  +\ldots+\left(  2(N-1)-\frac{1}%
{2}\right)  \right)  \\
&  =(1+o(1))\frac{M}{N}\left(  \frac{1}{2}+\frac{2N-1}{2}\cdot(N-1)\right)  \\
&  =(1+o(1))\frac{M}{2}\cdot\frac{(2N-1)(N-1)+1}{N}.
\end{align*}
Otherwise, if the robber is not caught, C just randomly checks rays: starting
from the center, C chooses a random ray, goes until the end of the ray,
returns to the center, and continues like this, until the robber is caught.
The expected capture time in this case is
\[
\sum_{j\geq1}\left(  (1-\frac{1}{N})^{j-1}\frac{1}{N}\left(
2(j-1)M+M/2\right)  \right)  =O(MN)=O(M).
\]
Since this happens with probability $O(M^{-1/3})$, the contribution of the
case where the robber switches rays is $o(M)$, and therefore for this strategy
of $C$, the expected capture time is
\[
(1+o(1))\frac{M}{2}\cdot\frac{(2N-1)(N-1)+1}{N}.
\]
\noindent\textbf{(ii.2) }Now suppose $C$ starts at the center of the ray,
rather than the end, and checks all rays from there. By the same arguments as
before, the capture time is%
\[
(1+o(1))\frac{M}{N}\left(  \frac{1}{2}+\left(  2+\frac{1}{2}\right)  +\left(
4+\frac{1}{2}\right)  +\ldots+\left(  2+2(N-2)+\frac{1}{2}\right)  \right)
\]
which is worse than in the case when starting at the end of a ray.

\noindent\textbf{(ii.3) }Similarly, suppose the cop starts at distance $cM$
from the center, for some $c\in\lbrack0,1]$. If he first goes to the center of
the ray, and then checks all rays (suppose the one he came from is the last to
be checked), then the capture time is
\begin{align*}
(1+o(1))\frac{M}{N} &  \left(  \frac{c^{2}}{2}+\left(  c+\frac{1}{2}\right)
+\left(  c+2+\frac{1}{2}\right)  +\ldots+\right.  \\
&  \left(  \left(  c+2(N-2)+\frac{1}{2}\right)  +(1-c)\left(  2c+2(N-1)+\frac
{1-c}{2}\right)  \right)  ,
\end{align*}
which is minimized for $c=1$. And if C goes first to the end of the ray, and
then to the center, the capture time is
\begin{align*}
(1+o(1))\frac{M}{N} &  \left(  \frac{((1-c)^{2}}{2}+c\left(  2(1-c)+\frac
{c}{2}\right)  +\left(  2(1-c)+c+\frac{1}{2}\right)  +\ldots+\right.  \\
&  \left(  \left(  2(1-c)+c+2(N-2)+\frac{1}{2}\right)  \right)  ,
\end{align*}
which for $N\geq2$ is also minimized for $c=1$ (in fact, for $N=2$ the numbers
are equal).

In short, the smallest capture time is achieved when C starts at the end of
some ray and therefore
\[
dct_{i}(S_{N,M})=(1+o(1))\frac{M}{2}\cdot\frac{(2N-1)(N-1)+1}{N}.
\]

\noindent\underline{\textbf{(iii) Computing }$H_{d}\left(  S_{N,M}\right)  $}.
It follows that for all $N\in\mathbb{N}-\{1\}$ we have
\[
H_{d}\left(  S_{N,M}\right)  =\frac{dct_{i}\left(  S_{N,M}\right)
}{dct\left(  S_{N,M}\right)  }=(1+o(1))\frac{(2N-1)(N-1)+1}{N}\geq2N-3,
\]
completing the proof.
\end{proof}

\subsection{Cost of Visibility in the Edge CR\ Games}

The cost of visibility in the edge CR\ games is defined analogously to that of
node games.

\begin{definition}
\label{prp0403}For every $G$, the \emph{edge adversarial cost of visibility}
is $\overline{H}_{a}(G)=\frac{\overline{ct}_{i}(G)}{\overline{ct}(G)}$ and the
\emph{edge drunk cost of visibility} is defined as $\overline{H}_{d}%
(G)=\frac{\overline{dct}_{i}(G)}{\overline{dct}(G)}$.
\end{definition}

Clearly, for every $G$ we have $\overline{H}_{a}\left(  G\right)  \geq1$ and
$\overline{H}_{d}\left(  G\right)  \geq1$. The following theorems show that in
fact both $\overline{H}_{a}\left(  G\right)  $ and $\overline{H}_{d}\left(
G\right)  $ can become arbitrarily large. To prove these theorems we will use
the previously introduced star graph $S_{N,1}$ and its line graph which is the
clique $K_{N}$. These graphs are illustrated in Figure \ref{fig004a} for
$N=6$. \begin{figure}[H]
\centering\scalebox{0.3}{\includegraphics{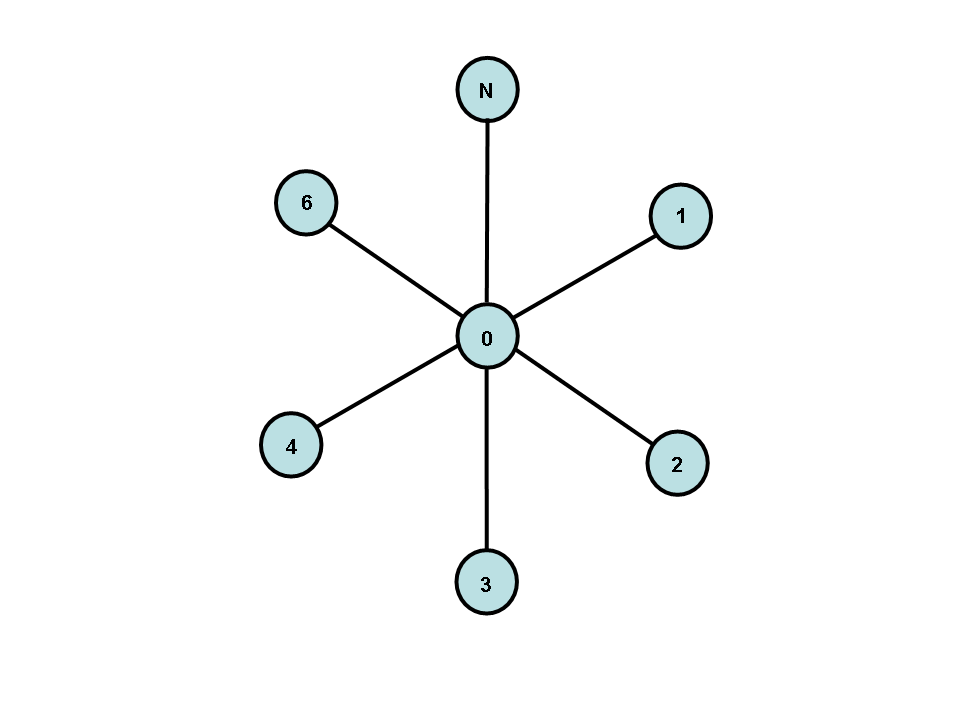}}\scalebox{0.3}{\includegraphics{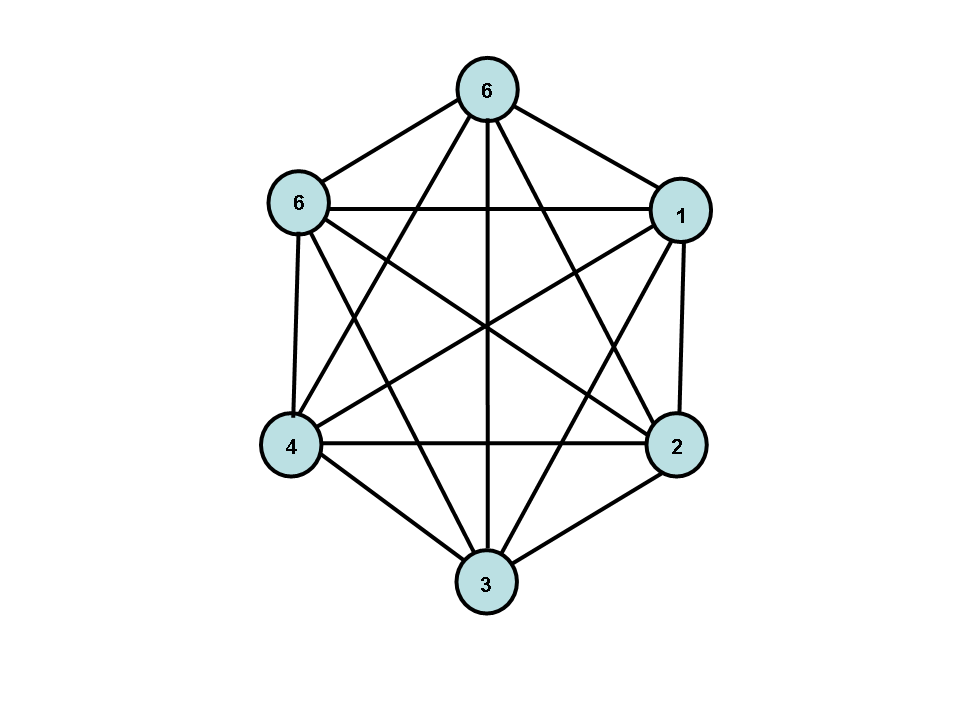}}\caption{The
\emph{star }graph $S_{6,1}$ and its line graph, the clique $K_{6}$.}%
\label{fig004a}%
\end{figure}

\begin{theorem}
\label{prp0404}For every $N\in\mathbb{N}-\{1\}$ we have $\overline{H}%
_{a}\left(  S_{N,1}\right)  =N-1$.
\end{theorem}

\begin{proof}
We have $\overline{H}_{a}(S_{N,1})=\frac{\overline{ct}_{i}(S_{N,1})}%
{\overline{ct}(S_{N,1})}=\frac{ct_{i}(K_{N})}{ct(K_{N})}$ and, since $N\geq2$,
clearly $ct(K_{N})=1$. Let us now compute $ct_{i}(K_{N})$.

For an upper bound on $ct_{i}(K_{N})$, C might just move to a random vertex.
If the robber stays still or if he moves to a vertex different from the one
occupied by C, he will be caught in the next step with probability $1/(N-1)$,
and thus an upper bound on the capture time is $N-1$.

For a lower bound, suppose that the robber always moves to a randomly chosen
vertex, different from the one occupied by $C$, and including the one occupied
by him now (that is, with probability $1/(N-1)$ he stands still, and after his
turn, he is with probability $1/(N-1)$ at each vertex different from the
vertex occupied by $C$. Hence $C$ is forced to move, and since he has no idea
where to go, the best strategy is also to move randomly, and the robber will
be caught with probability $1/(N-1)$, yielding a lower bound on the capture
time of $N-1$. Therefore
\[
ct_{i}\left(  K_{N}\right)  =N-1,
\]

Hence
\[
\overline{H}_{a}(S_{N,1})=\frac{\overline{ct}_{i}(S_{N,1})}{\overline
{ct}(S_{N,1})}=\frac{ct_{i}(K_{N})}{ct(K_{N})}=N-1.
\]

\end{proof}

\begin{theorem}
\label{prp0404b}For every $N\in\mathbb{N}-\{1\}$ we have $\overline{H}%
_{d}\left(  S_{N,1}\right)  =\frac{N(N-1)}{2N-3}$.
\end{theorem}

\begin{proof}
This is quite similar to the adversarial case. We have $\overline{H}%
_{d}(S_{N,1})=\frac{\overline{dct}_{i}(S_{N,1})}{\overline{dct}(S_{N,1}%
)}=\frac{dct_{i}(K_{N})}{dct(K_{N})}$. Clearly we have $dct(K_{N})=1-1/N$
(with probability $1/N$ the robber selects the same vertex to start with as
the cop and is caught before the game actually starts; otherwise is caught in
the first round).

For $dct_{i}(K_{N})$, it is clear that the strategy of constantly moving is
best for the cop, as in this case there are two chances to catch the robber
(either by moving towards him, or by afterwards the robber moving onto the
cop). It does not matter where he moves to as long as he keeps moving, and we
may thus assume that he starts at some vertex $v$ and moves to some other
vertex $w$ in the first round, then comes back to $v$ and oscillates like that
until the end of the game. When the cop moves to another vertex, the
probability that the robber is there is $1/(N-1)$. If he is still not caught,
the robber moves to a random place, thereby selecting the vertex occupied by
the cop with probability $1/(N-1)$. Hence, the probability to catch the robber
in one step is $\frac{1}{N-1}+(1-\frac{1}{N-1})\frac{1}{N-1}=\frac
{2N-3}{(N-1)^{2}}$. Thus, this time the capture time is a geometric random
variable with probability of success equal to $\frac{2N-3}{(N-1)^{2}}$. We get
$dct_{i}(K_{N})=\frac{(N-1)^{2}}{2N-3}$ and so
\[
\overline{H}_{d}(S_{N,1})=\frac{\overline{dct}_{i}(S_{N,1})}{\overline
{dct}(S_{N,1})}=\frac{dct_{i}(K_{N})}{dct(K_{N})}=\frac{(N-1)^{2}%
/(2N-3)}{(N-1)/N}=\frac{N(N-1)}{2N-3},
\]
which can become arbitrarily large by appropriate choice of $N$.
\end{proof}

\section{Algorithms for COV\ Computation\label{sec05}}

For graphs of relatively simple structure (e.g., paths, cycles, full trees,
grids) capture times and optimal strategies can be found by analytical
arguments \cite{kehagias2013cops,kehagias2012some}. For more complicated
graphs, an algorithmic solution becomes necessary. In this section we present
algorithms for the computation of capture time in the previously introduced
node CR variants. The same algorithms can be applied to the edge variants by
replacing $G$ with $L\left(  G\right)  $.

\subsection{Algorithms for Visible Robbers\label{sec0501}}

\subsubsection{Algorithm for Adversarial Robber\label{sec050101}}

The av-CR capture time $ct(G)$ can be computed in polynomial time. In fact,
stronger results have been presented by Hahn and MacGillivray; in
\cite{hahn2006note} they present an algorithm which, given $K$, computes for
every $\left(  x,y\right)  \in V^{2}$ the following:

\begin{enumerate}
\item $C\left(  x,y\right)  $, the optimal game duration when the cop/robber
configuration is $(x,y)$ and it is C's turn to play;

\item $R\left(  x,y\right)$, the optimal game duration when the cop/robber
configuration is $(x,y)$ and it is R's turn to play.\footnote{When $K<c(G)$,
there exist $\left(  x,y\right)  $ such that $C\left(  x,y\right)  =R\left(
x,y\right)  =\infty$; Hahn and MacGillivray's algorithm computes this
correctly, as well.}
\end{enumerate}

\noindent The av-CR capture time can be computed by $ct(G)=\min_{x\in V}%
\max_{y\in V}C\left(  x,y\right)  $; the optimal search strategies
$\widehat{\sigma}_{C}$, $\widehat{\sigma}_{R}$ can also be easily obtained
from the \emph{optimality equations}, as will be seen a little later. We have
presented in \cite{kehagias2012some} an implementation of Hahn and
MacGillivray's algorithm, which we call CAAR (Cops Against Adversarial
Robber). Below we present the algorithm for the case of a single cop (the
generalization for more than one cop is straightforward).

\newpage

\begin{center}
----------------------------------------------------------------------------------------------------

\underline{\textbf{The Cops Against Adversarial Robber (CAAR) Algorithm}}
\end{center}

\textbf{Input}: $G=(V,E)$

\texttt{01\qquad For All }$\left(  x,y\right)  \in V_{D}^{2}$

\texttt{02\qquad\qquad}$C^{\left(  0\right)  }\left(  x,y\right)  =0$

\texttt{03\qquad\qquad}$R^{\left(  0\right)  }\left(  x,y\right)  =0$

\texttt{04\qquad EndFor}

\texttt{05\qquad For All }$\left(  x,y\right)  \in V^{2}-V_{D}^{2}$

\texttt{06\qquad\qquad}$C^{\left(  0\right)  }\left(  x,y\right)  =\infty$

\texttt{07\qquad\qquad}$R^{\left(  0\right)  }\left(  x,y\right)  =\infty$

\texttt{08\qquad EndFor}

\texttt{09\qquad}$i=1$

\texttt{10\qquad While }$1>0$

\texttt{11\qquad\qquad For All }$\left(  x,y\right)  \in V^{2}-V_{D}^{2}$

\texttt{12\qquad\qquad\qquad}$C^{\left(  i\right)  }\left(  x,y\right)
=1+\min_{x^{\prime}\in N\left[  x\right]  }R^{\left(  i-1\right)  }\left(
x^{\prime},y\right)  $

\texttt{13\qquad\qquad\qquad}$R^{\left(  i\right)  }\left(  x,y\right)
=1+\max_{y^{\prime}\in N\left[  y\right]  }C^{\left(  i\right)  }\left(
x,y^{\prime}\right)  $

\texttt{14\qquad\qquad EndFor}

\texttt{15\qquad\qquad If }$C^{\left(  i\right)  }=C^{\left(  i-1\right)  }%
$\texttt{ And}$\ R^{\left(  i\right)  }=R^{\left(  i-1\right)  }$

\texttt{16\qquad\qquad\qquad Break}

\texttt{17\qquad\qquad EndIf }

\texttt{18\qquad\qquad}$i\leftarrow i+1$

\texttt{19\qquad EndWhile}

\texttt{20\qquad}$C=C^{\left(  i\right)  }$

\texttt{21\qquad}$R=R^{\left(  i\right)  }$

\textbf{Output}\texttt{: }$C$\texttt{, }$R$

\begin{center}
----------------------------------------------------------------------------------------------------
\end{center}

The algorithm operates as follows. In lines \texttt{01-08} $C^{\left(
0\right)  }\left(  x,y\right)  $ and $R^{\left(  0\right)  }\left(
x,y\right)  $ are initialized to $\infty$, except for \textquotedblleft
diagonal\textquotedblright\ positions $\left(  x,y\right)  \in V_{D}^{2}$
(i.e., positions with $x=y$) for which we obviously have $C\left(  x,x\right)
=R\left(  x,x\right)  =0$. Then a loop is entered (lines \texttt{10-19}) in
which $C^{\left(  i\right)  }\left(  x,y\right)  $ is computed (line
\texttt{12}) by letting the cop move to the position which achieves the
smallest capture time (according to the currently available estimate
$R^{\left(  i-1\right)  }\left(  x,y\right)  $); $R^{\left(  i\right)
}\left(  x,y\right)  $ is computed similarly in line \texttt{13}, looking for
the largest capture time. This process is repeated until no further changes
take place, at which point the algorithm exits the loop and 
terminates\footnote{This algorithm is a game theoretic version of \emph{value iteration}
\cite{puterman1994markov}, which we see again in Section \ref{sec0502}.}. 
It has been proved in \cite{hahn2006note} that, for any graph $G$ and any
$K\in\mathbb{N}$, CAAR\ always terminates and the finally obtained $\left(
C,R\right)  $ pair satisfies the \emph{optimality equations}
\begin{align}
\forall\left(  x,y\right)   &  \in V_{D}^{2}:C\left(  x,y\right)
=0;\quad\forall\left(  x,y\right)  \in V^{2}-V_{D}^{2}:C\left(  x,y\right)
=1+\min_{x^{\prime}\in N\left[  x\right]  }R\left(  x^{\prime},y\right)
,\label{eq0502}\\
\forall\left(  x,y\right)   &  \in V_{D}^{2}:R\left(  x,y\right)
=0;\quad\forall\left(  x,y\right)  \in V^{2}-V_{D}^{2}:R\left(  x,y\right)
=1+\max_{y^{\prime}\in N\left[  y\right]  }C\left(  x,y^{\prime}\right)
.\label{eq0504}%
\end{align}
The optimal memoryless strategies $\sigma_{C}^{\left(  K\right)  }\left(
x,y\right)  $, $\sigma_{R}^{\left(  K\right)  }\left(  x,y\right)  $ can be
computed for every position $(x,y)$ by letting $\sigma_{C}^{\left(  K\right)}\left(  x,y\right)$ 
(resp. $\sigma_{R}^{\left(  K\right)  }\left(
x,y\right)  $ ) be a node $x^{\prime}\in N\left[  x\right]$ 
(resp. $y^{\prime}\in N\left[  y\right]$) which achieves the minimum in
(\ref{eq0502}) (resp. maximum in (\ref{eq0504})). The capture time $ct(G)$ is
computed from%
\[
ct\left(  G\right)  =\min_{x\in V}\max_{y\in V}C\left(  x,y\right)  .
\]

\subsubsection{Algorithm for Drunk Robber\label{sec050102}}

For any given $K$, \emph{value iteration}\ can be used to determine both
$dct\left(  G,K\right)  $ and the optimal strategy $\sigma_{C}^{\left(
K\right)  }\left(  x,y\right)  $; one implementation is our CADR\ (Cops
Against Drunk Robber) algorithm \cite{kehagias2012some} which is a typical
value-iteration \cite{puterman1994markov} MDP algorithm; alternatively, CADR
can be seen as an extension of the CAAR idea to the dv-CR. Below we present
the algorithm for the case of a single cop (the generalization for more than
one cops is straightforward).

\begin{center}
----------------------------------------------------------------------------------------------------

\underline{\textbf{The Cops Against Drunk Robber (CADR) Algorithm}}
\end{center}

\textbf{Input}: $G=(V,E)$, $\varepsilon$

\texttt{01 For All }$\left(  x,y\right)  \in V_{D}^{2}$

\texttt{02 } $\ C^{\left(  0\right)  }\left(  x,y\right)  =0$

\texttt{03 EndFor}

\texttt{04 For All }$\left(  x,y\right)  \in V-V_{D}^{2}$

\texttt{05 } $\ C^{\left(  0\right)  }\left(  x,y\right)  =\infty$

\texttt{06 EndFor}

\texttt{07 }$i=1$

\texttt{08 While }$1>0$

\texttt{09\qquad For All }$\left(  x,y\right)  \in V-V_{D}^{2}$

\texttt{10\qquad\qquad}$C^{\left(  i\right)  }\left(  x,y\right)
=1+\min_{x^{\prime}\in N\left[  x\right]  }\sum_{y^{\prime}\in V}P\left(
\left(  x^{\prime},y\right)  \rightarrow\left(  x^{\prime},y^{\prime}\right)
\right)  C^{\left(  i-1\right)  }\left(  x^{\prime},y^{\prime}\right)  $

\texttt{11\qquad EndFor}

\texttt{12\qquad If }$\max_{\left(  x,y\right)  \in V^{2}}\left\vert
C^{\left(  i\right)  }\left(  x,y\right)  -C^{\left(  i-1\right)  }\left(
x,y\right)  \right\vert <\varepsilon$

\texttt{13\qquad\qquad Break}

\texttt{14\qquad EndIf }

\texttt{15\qquad}$i\leftarrow i+1$

\texttt{16 EndWhile}

\texttt{17 }$C=C^{\left(  i\right)  }$

\textbf{Output}\texttt{: }$C$

\begin{center}
----------------------------------------------------------------------------------------------------
\end{center}

The algorithm operates as follows (again we use $C\left(  x,y\right)  $ to
denote the optimal expected game duration when the game position is $\left(
x,y\right)  $). In lines \texttt{01-06} $C^{\left(  0\right)  }\left(
x,y\right)  $ is initialized to $\infty$, except for \textquotedblleft
diagonal\textquotedblright positions $\left(  x,y\right) \in V_{D}^{2}$. In
the main loop (lines \texttt{08-16}) $C^{\left(  i\right)  }\left(
x,y\right)  $ is computed (line \texttt{10}) by letting the cop move to the
position which achieves the smallest expected capture time ($P\left(  \left(
x^{\prime},y\right)  \rightarrow\left(  x^{\prime},y^{\prime}\right)  \right)
$ in line \texttt{10} indicates the transition probability from$\left(
x^{\prime},y\right)  $ to $\left(  x^{\prime},y^{\prime}\right)  $). This
process is repeated until the maximum change $\left\vert C^{\left(  i\right)
}\left(  x,y\right)  -C^{\left(  i-1\right)  }\left(  x,y\right)  \right\vert
$ is smaller than the \emph{termination criterion} $\varepsilon$, at which
point the algorithm exits the loop and terminates. This is a typical
\emph{value iteration} MDP algorithm \cite{puterman1994markov}; the
convergence of such algorithms has been studied by several authors, in various
degrees of generality \cite{de1980optimal,eaton1962optimal,howard1971dynamic}.
A simple yet strong result, derived in \cite{eaton1962optimal}, uses the
concept of \emph{proper} strategy: a strategy is called proper if it yields
finite expected capture time. It is proved in \cite{eaton1962optimal} that, if
a proper strategy exists for graph G, then CADR-like algorithms converge. 
In the case of dv-CR we know that C has a proper strategy: it is the
random walking strategy $\overline{s}^{(K)}_C$ mentioned in Theorem \ref{prp0307}.
Hence CADR converges and in the limit, $C=\lim_{i\rightarrow\infty}C^{\left(  i\right)  }$ 
satisfies the \emph{optimality equations}
\begin{equation}
\forall\left(  x,y\right)  \in V_{D}^{2}:C\left(  x,y\right)  =0;\quad
\forall\left(  x,y\right)  \in V^{2}-V_{D}^{2}:C\left(  x,y\right)
=1+\min_{x^{\prime}\in N\left[  x\right]  }\sum P\left(  \left(  x^{\prime
},y\right)  \rightarrow\left(  x^{\prime},y^{\prime}\right)  \right)  C\left(
x^{\prime},y^{\prime}\right)  .
\end{equation}
The optimal memoryless strategy $\sigma_{C}^{\left(  K\right)  }\left(
x,y\right)  $ can be computed for every position $(x,y)$ by letting
$\sigma_{C}^{\left(  K\right)  }\left(  x,y\right)  $ be a node $x^{\prime}\in
N\left[  x\right]  $ (resp. $y^{\prime}\in N\left[  y\right]$) which
achieves the minimum in (\ref{eq0502}) (resp. maximum in (\ref{eq0504})). The
capture time $dct(G)$ is computed from%
\[
dct\left(  G\right)  =\min_{x\in V}C\left(  x,y\right)  .
\]

\subsection{Algorithms for Invisible Robbers\label{sec0502}}

\subsubsection{Algorithms for Adversarial Robber\label{sec050201}}

We have not been able to find an efficient algorithm for solving the ai-CR
game. Several algorithms for \emph{imperfect information stochastic games}
could be used to this end but we have found that they are practical only for
very small graphs.

\subsubsection{Algorithm for Drunk Robber\label{sec050202}}

In the case of the drunk invisible robber we are also using a game tree search
algorithm with pruning, for which some analytical justification can be
provided. We call this the \emph{Pruned Cop Search} (PCS) algorithm. Before
presenting the algorithm we will introduce some notation and then prove a
simple fact about expected capture time. We limit ourselves to the single cop
case, since the extension to more cops is straightforward.

We let $\mathbf{x}=x_{0}x_{1}x_{2}\ldots\ $ be an infinite history of cop
moves. Letting $t$ being the current time step, the probability vector
$\mathbf{p}\left(  t\right)  $ contains the probabilities of the robber being
in node $v\in V$ or in the \emph{capture state }$n+1$; more
specifically:\ $\mathbf{p}\left(  t\right)  =\left[  p_{1}\left(  t\right)
,\ldots,p_{v}\left(  t\right)  ,\ldots,p_{n}\left(  t\right)  ,p_{n+1}\left(
t\right)  \right]  $ and $p_{v}\left(  t\right)  =\Pr\left(  y_{t}%
=v|x_{0}x_{1}\ldots x_{t}\right)  $. Hence $\mathbf{p}\left(  t\right)  $
depends (as expected) on the \emph{finite} cop history $x_{0}x_{1}\ldots
x_{t}$. The expected capture time is denoted by $C\left(  \mathbf{x}\right)
=E(T|\mathbf{x})$; the conditioning is on the infinite cop history. The
PCS\ algorithm works because $E\left(  T|\mathbf{x}\right)  $ can be
approximated from a finite part of $\mathbf{x}$, as explained below. We have
\begin{equation}
C\left(  \mathbf{x}\right)  =E\left(  T|\mathbf{x}\right)  =\sum_{t=0}%
^{\infty}t\cdot\Pr\left(  T=t|\mathbf{x}\right)  =\sum_{t=0}^{\infty}%
\Pr\left(  T>t|\mathbf{x}\right)  .\label{eq0511}%
\end{equation}
$\mathbf{x}$\ in the conditioning is the \emph{infinite} history
$\mathbf{x}=x_{0}x_{1}x_{2}\ldots$ . However, for every $t$ we have%
\[
\Pr\left(  T>t|\mathbf{x}\right)  =1-\Pr\left(  T\leq t|\mathbf{x}\right)
=1-\Pr\left(  T\leq t|x_{0}x_{1}\ldots x_{t}\right)  .
\]
Let us define
\[
C^{\left(  t\right)  }\left(  x_{0}x_{1}\ldots x_{t}\right)  =\sum_{\tau
=0}^{t}\left[  1-\Pr\left(  T\leq\tau|x_{0}x_{1}\ldots x_{\tau}\right)
\right]  =\sum_{\tau=0}^{t}\left[  1-p_{n+1}\left(  \tau\right)  \right]  ,
\]
where $p_{n+1}\left(  \tau\right)  $ is the probability that the robber is in
the \emph{capture state} $n+1$ at time $\tau$ (the dependence on $x_{0}%
x_{1}\ldots x_{\tau}$ is suppressed, for simplicity of notation). Then for
all $t$ we have
\begin{equation}
C^{\left(  t\right)  }\left(  x_{0}x_{1}\ldots x_{t}\right)  =C^{\left(
t-1\right)  }\left(  x_{0}x_{1}\ldots x_{t-1}\right)  +\left(  1-p_{n+1}%
\left(  t\right)  \right)  .\label{eq0512}%
\end{equation}
Update (\ref{eq0512}) can be computed using only the previous cost 
$C^{\left(t-1\right)  }\left(  x_{0}x_{1}\ldots x_{\tau-1}\right)  $ and the
(previously computed)\ probability vector $\mathbf{p}\left(t\right)  $.
While $C^{\left(t\right)  }\left(  x_{0}\ldots x_{t}\right)  \leq
C\left(  \mathbf{x}\right)  $, we hope that (at least for the
\textquotedblleft good\textquotedblright\ histories) we have
\begin{equation}
\lim_{t\rightarrow\infty}C^{\left(  t\right)  }\left(  x_{0}\ldots
x_{t}\right)  =C\left(  \mathbf{x}\right)  .\label{eq0513}%
\end{equation}
This actually works well in practice.

The PCS algorithm is given below in pseudocode. We have introduced a structure
$S$ with fields $S.\mathbf{x}$, $S.\mathbf{p}$, $S.C=C\left(  S.\mathbf{x}%
\right)  $. Also we denote concatenation by the $|$ symbol, i.e., $x_{0}%
x_{1}\ldots x_{t}|v=x_{0}x_{1}\ldots x_{t}v$.

\begin{center}
----------------------------------------------------------------------------------------------------

\underline{\textbf{The Pruned Cop Search (PCS) Algorithm}}
\end{center}

\textbf{Input}: $G=(V,E)$\texttt{, }$x_{0}$\texttt{, }$J_{max}$\texttt{,
}$\varepsilon$

\texttt{01 \ }$t=0$

\texttt{02 \ }$S.\mathbf{x}=x_{0}$, $S.\mathbf{p}=\Pr(y_{0}|x_{0})$, $S.C=0$

\texttt{03 \ }$\mathbf{S}=\{S\}$

\texttt{04 \ }$C_{best}^{old}=0$

\texttt{05 \ While }$1>0$

\texttt{06\qquad\qquad}$\widetilde{\mathbf{S}}=\emptyset$

\texttt{07\qquad\qquad For All }$S\in\mathbf{S}$

\texttt{08\qquad\qquad\qquad}$\mathbf{x}=S.\mathbf{x}$\texttt{, }%
$\mathbf{p}=S.\mathbf{p}$\texttt{, }$C=S.C$

\texttt{09\qquad\qquad\qquad For All }$v\in N\left[  x_{t}\right]  $

\texttt{10\qquad\qquad\qquad\qquad}$\mathbf{x}^{\prime}=\mathbf{x}|v$

\texttt{11\qquad\qquad\qquad\qquad}$\mathbf{p}^{\prime}=\mathbf{p}\cdot P(v)$

\texttt{12\qquad\qquad\qquad\qquad}$C^{\prime}=\mathbf{Cost}(\mathbf{x}%
^{\prime},\mathbf{p}^{\prime},C)$

\texttt{13\qquad\qquad\qquad\qquad}$S^{\prime}.\mathbf{x}=\mathbf{x}^{\prime}%
$\texttt{, }$S^{\prime}.\mathbf{p}=\mathbf{p}^{\prime}$\texttt{, }$S^{\prime
}.C=C^{\prime}$

\texttt{14\qquad\qquad\qquad\qquad}$\widetilde{\mathbf{S}}=\widetilde
{\mathbf{S}}\cup\{S^{\prime}\}$

\texttt{15\qquad\qquad\qquad EndFor}

\texttt{16\qquad\qquad EndFor}

\texttt{17\qquad\qquad}$\mathbf{S}=\mathbf{Prune}(\widetilde{\mathbf{S}%
},J_{max})$

\texttt{18\qquad\qquad}$[\mathbf{x}_{best},C_{best}]=\mathbf{Best}%
(\mathbf{S})$

\texttt{19\qquad\qquad If }$|C_{best}-C_{best}^{old}|<\varepsilon$

\texttt{20\qquad\qquad\qquad Break}

\texttt{21\qquad\qquad Else}

\texttt{22\qquad\qquad\qquad}$C_{best}^{old}=C_{best}$

\texttt{23\qquad\qquad\qquad}$t\leftarrow t+1$

\texttt{24\qquad\qquad EndIf}

\texttt{25 \ EndWhile}

\textbf{Output}\texttt{: }$\mathbf{x}_{best}$\texttt{, }$C_{best}=C\left(
\mathbf{x}_{best}\right)  $\texttt{.}

\begin{center}
----------------------------------------------------------------------------------------------------
\end{center}

The PCS\ algorithm operates as follows. At initialization (lines
\texttt{01-04}), we create a single $S$ structure (with $S.\mathbf{x}$ being
the initial cop position, $S.\mathbf{p}$ the initial, uniform robber
probability and $S.C=0$) which we store in the set $\mathbf{S}$. Then we enter
the main loop (lines \texttt{05-25}) where we pick each available cop sequence
$\mathbf{x}$ of length $t$ (line \texttt{08}). Then, in lines \texttt{09-15}
we compute, for all legal extensions $\mathbf{x}^{\prime}=\mathbf{x}|v$ (where
$v\in N\left[  x_{t}\right]  $) of length $t+1$ (line \texttt{10}), the
corresponding $\mathbf{p}^{\prime}$ (line \texttt{11})\ and $C^{\prime}$ (by
the subroutine $\mathbf{Cost}$ at line \texttt{12}). We store these quantities
in $S^{\prime}$ which is placed in the temporary storage set $\widetilde
{\mathbf{S}}$ (lines \texttt{13-14}). After exhausting all possible extensions
of length $t+1$, we prune the temporary set $\widetilde{\mathbf{S}}$,
retaining only the $J_{\max}$ best cop sequences (this is done in line
\texttt{17} by the subroutine \textbf{Prune} which computes \textquotedblleft
best\textquotedblright\ in terms of smallest $C\left(  \mathbf{x}\right)  $).
Finally, the subroutine $\mathbf{Best}$ in line \texttt{18} computes the
overall smallest expected capture time $C_{best}=C\left(  \mathbf{x}%
_{best}\right)  $. The procedure is repeated until the termination criterion
$|C_{best}-C_{best}^{old}|<\varepsilon$ is satisfied. As explained above, the
criterion is expected to be always eventually satisfied because of
(\ref{eq0513}).

\section{Experimental Estimation of The Cost of Visibility\label{sec06}}

We now present numerical computations of the drunk cost of visibility for
graphs which are not amenable to analytical computation\footnote{We do not
deal with the adversarial cost of visibility because, while we can compute
$ct\left(  G\right)  $ with the CAAR algorithm, we do not have an efficient
algorithm to compute $ct_{i}\left(  G\right)  $; hence we cannot perform
experiments on $H_{a}\left(  G\right)  =\frac{ct_{i}\left(  G\right)
}{ct\left(  G\right)}$. The difficulty with $ct_{i}\left(  G\right)  $ is
that ai-CR is a stochastic game of imperfect information; even for very small
graphs, one cop and one robber, ai-CR\ involves a state space with size far
beyond the capabilities of currently available stochastic games algorithms
(see \cite{raghavan1991algorithms}).}. In Section \ref{sec0601} we deal with
node games and in Section \ref{sec0602}  with edge games.

\subsection{\noindent Experiments with Node Games\label{sec0601}}

Since $H_{d}\left(  G\right)  =\frac{dct_{i}\left(  G\right)  }{dct\left(G\right)  }$, 
we use the CADR algorithm  to compute $dct\left(  G\right)$ 
and the PCS algorithm to compute $dct_{i}\left(  G\right)$. 
We use graphs $G$ obtained from
\emph{indoor environments}, which we represent by their \emph{floorplans}. In
Fig. \ref{fig060100} we present a floorplan and its graph representation. The
graph is obtained by decomposing the floorplan into convex \emph{cells},
assigning each cell to a node and connecting nodes by edges whenever the
corresponding cells are connected by an open space.

We have written a script which, given some parameters, generates random
floorplans and their graphs. Every floorplan consists of a rectangle divided
into orthogonal \textquotedblleft rooms\textquotedblright. If each internal room were
connected to its four nearest neighbors we would get an $M\times N$ grid
$G^{\prime}$. However, we randomly generate a spanning tree $G_T$ of $G^{\prime}$ 
and initially introduce doors only between rooms which are connected in
$G_T$. Our final graph $G$ is obtained from $G_T$ by iterating over all missing
edges and adding each one with probability $p_{0}\in\left[  0,1\right]  $.
Hence each floorplan is characterized by three parameters:\ $M$, $N$ and
$p_{0}$.

\begin{figure}[H]
\centering\scalebox{0.5}{\includegraphics{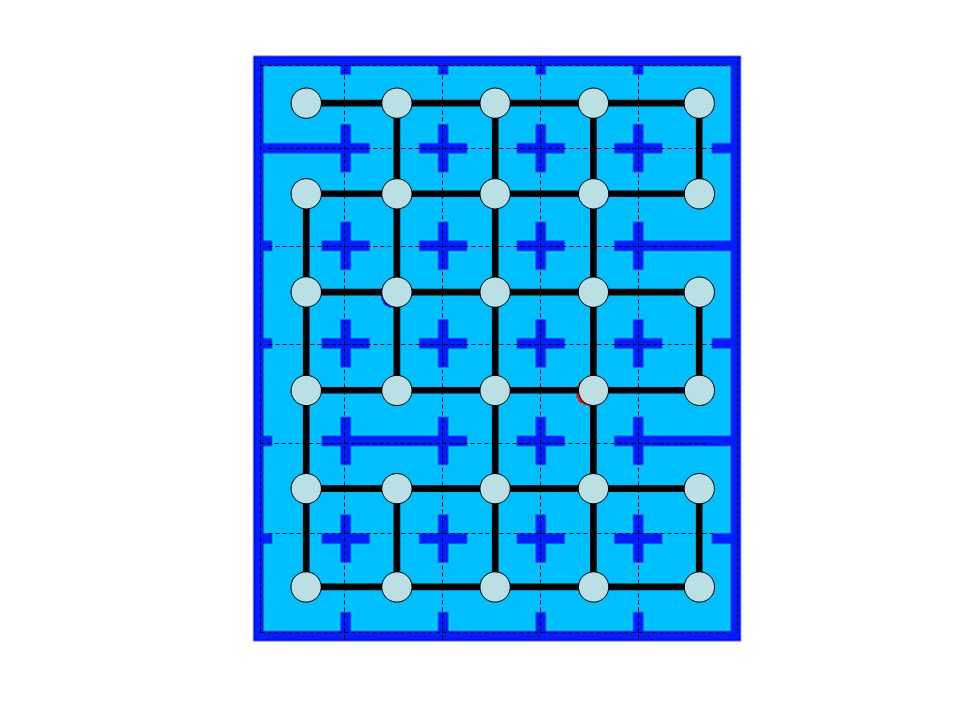}} \caption[A floorplan
and the corresponding graph.]{A floorplan and the corresponding graph.}%
\label{fig060100}%
\end{figure}

We use the following pairs of $\left(  M,N\right)  $ values:\ (1,30), (2,15),
(3,10), (4,7), (5,6). Four of these pairs give a total of 30 nodes and the pair 
($M=4$, $N=7$) gives $n=28$ nodes;  as $M/N$ increases, we progress from a
path to a nearly square grid. For each $\left(  M,N\right)  $ pair we use five
$p_{0}$ values: 0.00, 0.25, 0.50, 0.75, 1.00; note the
progression from a tree ($p_{0}=0.00$)\ to a full grid ($p_{0}=1.00$). For
each triple $\left(  M,N,p_{0}\right)  $ we generate 50 floorplans, obtain
their graphs and for each graph $G$ we compute $dct(G)$ using CADR,
$dct_{i}\left(  G\right)  $ using PCS and $H_{d}\left(  G\right)
=\frac{dct_{i}\left(  G\right)  }{dct\left(  G\right)  }$; finally we average
$H_{d}\left(  G\right)  $ over the 100 graphs. In Fig. \ref{fig060101} we plot
$dct(G)$ as a function of the probability $p_{0}$; each plotted curve
corresponds to an $\left(  M,N\right)  $ pair. Similarly, in Fig.
\ref{fig060102} we plot $dct_{i}(G)$ and in Fig. \ref{fig060103} we plot
$H_{d}(G)$.

\newpage

\begin{figure}[H]
\centering
\scalebox{0.5}{\includegraphics{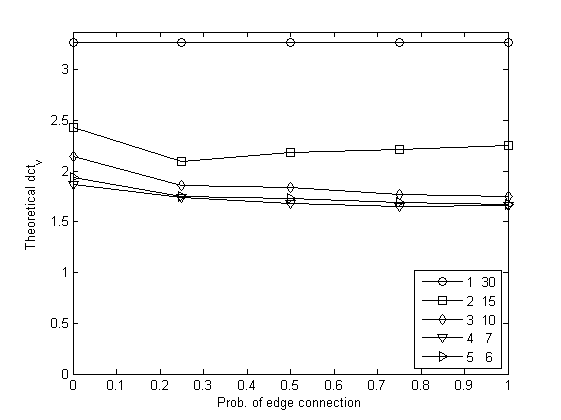}} \caption[$dct(G)$ curves for
floorplans with $n$=30 or $n$=28 cells.]{$dct(G)$ curves for floorplans with $n$=30 or $n$=28
cells. Each curve corresponds to a fixed $(M,N)$ pair. The horizontal axis
corresponds to the edge insertion probability $p_{0}$.}%
\label{fig060101}%
\end{figure}

\begin{figure}[H]
\centering
\scalebox{0.5}{\includegraphics{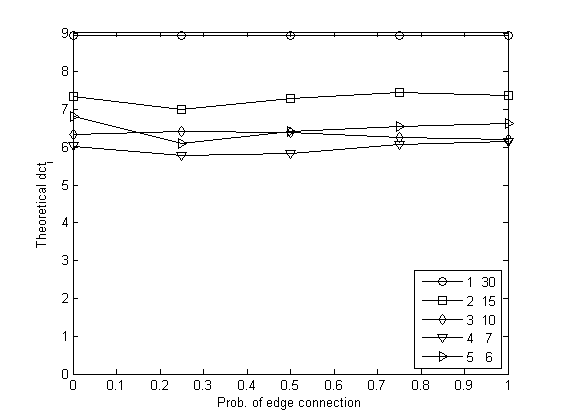}}\caption[$dct_{i}(G)$ curves
for floorplans with $n$=30 or $n$=28 cells.]{$dct_{i}(G)$ curves for floorplans with
$n$=30 or $n$=28 cells. Each curve corresponds to a fixed $(M,N)$ pair. The horizontal
axis corresponds to the edge insertion probability $p_{0}$.}%
\label{fig060102}%
\end{figure}

\begin{figure}[H]
\centering
\scalebox{0.5}{\includegraphics{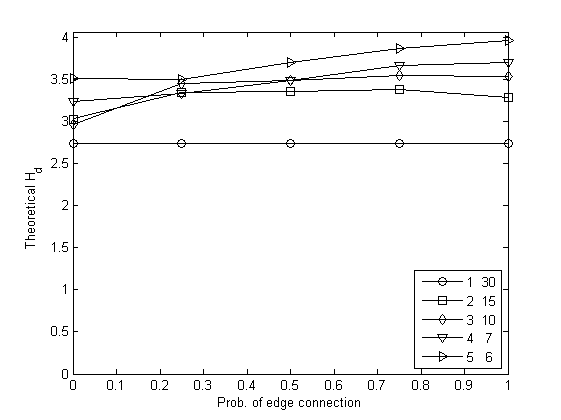}}\caption[$H_{d}(G)$ curves for
floorplans with $n$=30 or $n$=28 cells.]{$H_{d}(G)$ curves for floorplans with $n$=30 or $n$=28
cells. Each curve corresponds to a fixed $(M,N)$ pair. The horizontal axis
corresponds to the edge insertion probability $p_{0}$.}%
\label{fig060103}%
\end{figure}

\newpage

We can see in Fig. \ref{fig060101}-\ref{fig060102} that both $dct\left(
G\right)  $ and $dct_{i}\left(  G\right)$  are usually decreasing functions of 
the $M/N$ ratio. However the cost of visibility $H_{d}\left(  G\right)  $
\emph{increases} with $M/N$. This is due to the fact that, when the $M/N$ ratio is low,
$G$ is closer to a path and there is less difference in the search schedules
and capture times between dv-CR and di-CR. On the other hand, for high $M/N$
ratio, $G$ is closer to a grid, with a significantly increased ratio of 
edges to nodes  (as compared to the low $M/N$, path-like instances).
This, combined with the loss of information (visibility),
results in $H_d(G)$ being an increasing function of $M/N$.
The increase of $H_{d}\left(  G\right)$ with $p_{0}$ can be explained
in the same way, since increasing $p_0$ implies more edges and this makes the cops' task harder.

\subsection{\noindent Experiments with Edge Games\label{sec0602}}

Next we deal with 
$\overline{H}_{d}\left(  G\right)  =\frac{\overline{dct}_{i}\left(  G\right)  }{\overline{dct}\left(  G\right)  }$. 
We use graphs $G$
obtained from \emph{mazes} such as the one illustrated in Fig. \ref{fig060200}. 
Every \emph{corridor} of the maze corresponds to an edge; corridor
intersections correspond to nodes. The resulting graph $G$ is also depicted in
Fig.\ref{fig060200}. From $G$ we obtain the line graph $L(G)$, to which we
apply CADR\ to compute $dct\left(  L(G)\right)  =\overline{dct}\left(G\right)$ 
and PCS to compute  
$dct_{i}\left(  L( G)\right)  =\overline{dct}_{i}\left(  G\right)$.

\begin{figure}[H]
\centering\scalebox{0.5}{\includegraphics{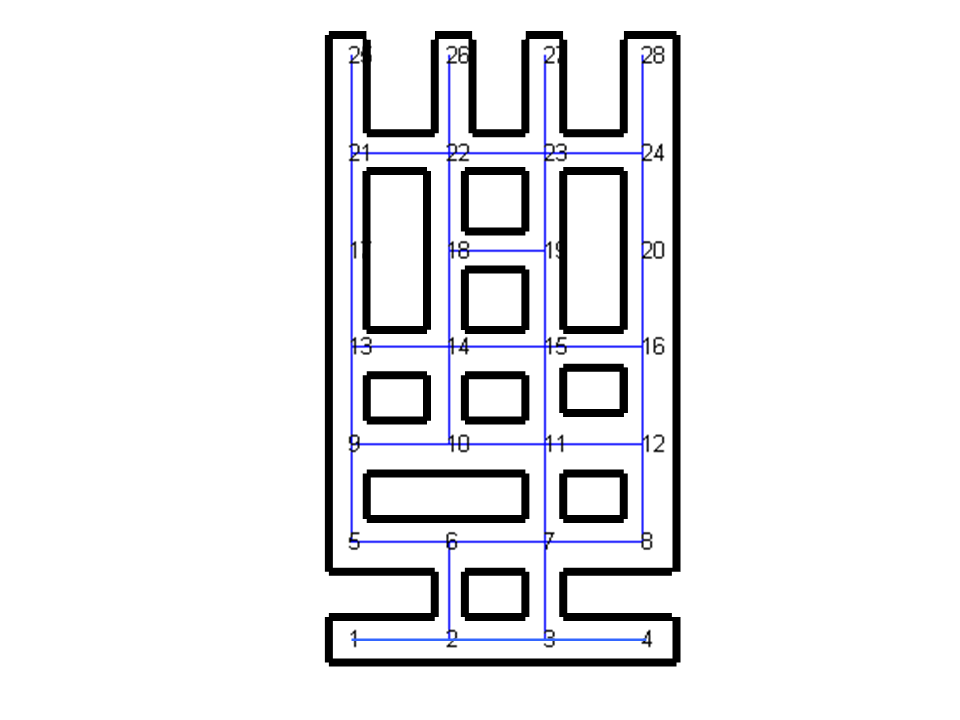}} \caption[A maze and
the corresponding graph.]{A maze and the corresponding graph.}%
\label{fig060200}%
\end{figure}

We use graphs of the same type as the ones of Section \ref{sec0601} but we now
focus on the edge-to-edge movements of cops and robber.
Hence from every $G$ (obtained by a specific $(M,N,p_0)$ triple) we produce
the line graph $L(G)$, for which we compute $H_d(L(G))$ using the CADR and PCS algorithms.
Once again we generate 50 graphs and present average $dct(G)$, $dct_{i}\left(  G\right)  $
and $H_{d}\left(  G\right)  $ results in Figures \ref{fig060201}-\ref{fig060203}. 
These figures are rather similar to Figures \ref{fig060101}-\ref{fig060103}, except that
the increase of $\overline{H}_{d}\left(  G\right)$ as a function of $M/N$ is greater 
than that of $H_{d}\left(  G\right)$. This is due to the fact that $L(G)$ has more 
nodes and edges than $G$, hence the loss of visibility makes the edge game 
significantly harder than the node game. There is one exception to the above remarks, 
namely the case $(M,N)=(1,30)$; in this case both $G$ and $L(G)$ are paths and 
$H_{d}\left(  G\right)$ is essentially equal to $\overline{H}_{d}\left(G\right)$
(as can be seen by comparing Figures \ref{fig060103} and \ref{fig060203}).

\newpage

\begin{figure}[H]
\centering
\scalebox{0.5}{\includegraphics{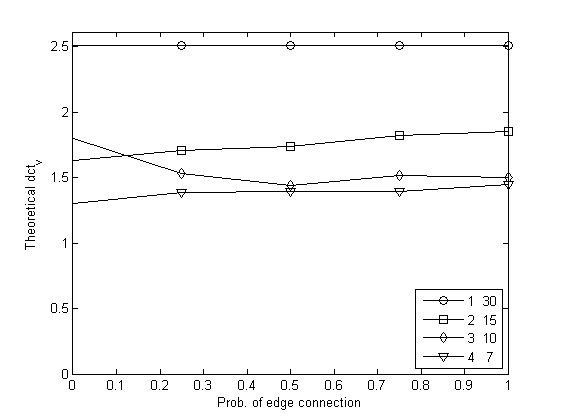}}\caption[$\overline{dct}(G)$
curves for floorplans with $n$=30 or $n$=28 cells.]{$\overline{dct}(G)$ curves for floorplans with
$n$=30 or $n$=28 cells. Each curve corresponds to a fixed $(M,N)$ pair. The horizontal
axis corresponds to the edge insertion probability $p_{0}$.}%
\label{fig060201}%
\end{figure}

\begin{figure}[H]
\centering
\scalebox{0.5}{\includegraphics{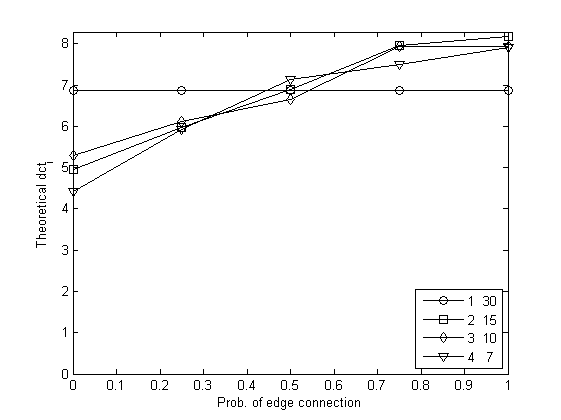}}\caption[$\overline{dct}%
_{i}(G)$ curves for floorplans with $n$=30 or $n$=28 cells.]{$\overline{dct}_{i}(G)$ curves for
floorplans with $n$=30 or $n$=28 cells. Each curve corresponds to a fixed $(M,N)$ pair.
The horizontal axis corresponds to the edge insertion probability $p_{0}$.}%
\label{fig060202}%
\end{figure}

\begin{figure}[H]
\centering
\scalebox{0.5}{\includegraphics{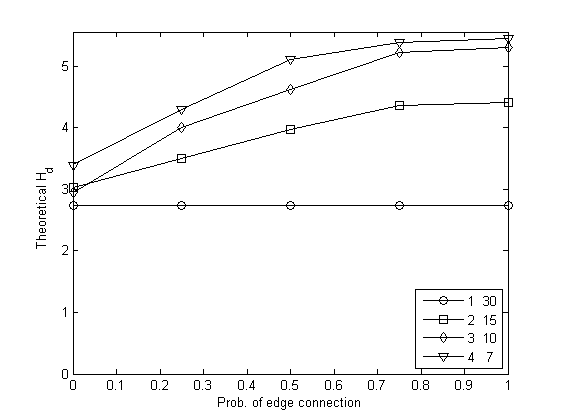}}\caption[$\overline{H}_{d}(G)$
curves for floorplans with $n$=30 or $n$=28 cells.]{$\overline{H}_{d}\left(G\right)$ curves for floorplans
with $n$=30 or $n$=28 cells. Each curve corresponds to a fixed $(M,N)$ pair. The
horizontal axis corresponds to the edge insertion probability $p_{0}$.}%
\label{fig060203}%
\end{figure}

\newpage

\section{Conclusion\label{sec07}}

In this paper we have studied two versions of the cops and robber game: the
one is played on the nodes of a graph and the other played on the edges. For
each version, we studied four variants, obtained by changing the visibility
and adversariality assumptions regarding the robber; hence we have a total of
eight CR games. For each of these we have defined \emph{rigorously} the
corresponding optimal capture time, using game theoretic and probabilistic tools.

Then, for the node games we have introduced the adversarial cost of visibility
$H\left(  G\right)  =\frac{ct_{i}\left(  G\right)  }{ct\left(  G\right)  }$
and the drunk cost of visibility $H_{d}\left(  G\right)  =\frac{dct_{i}\left(
G\right)  }{dct\left(  G\right)  }$ . These ratios quantify the increase in
difficulty of the CR\ game when the cop is no longer aware of the robber's
position (this situation occurs often in mobile robotics). 

We have defined analogous quantities ($\overline{H}\left(  G\right)  =\frac
{\overline{{ct}_{i}}\left(  G\right)  }{\overline{ct}\left(  G\right)  }$,
$\overline{H}_{d}\left(  G\right)  =\frac{\overline{dct}_{i}\left(
G\right)  }{\overline{dct}\left(  G\right)  }$) for the edge CR games.

We have studied analytically $H\left(  G\right)  $ and $H_{d}\left(  G\right)
$ and have established that both can get arbitrarily large. We have
established similar results for $\overline{H}\left(  G\right)  $ and
$\overline{H}_{d}\left(  G\right)  $. In addition, we have studied
$H_{d}\left(  G\right)  $ and $\overline{H}_{d}\left(  G\right)  $ by
numerical experiments which support both the game theoretic results of the current paper
and the analytical computations of capture times presented in \cite{kehagias2012some,kehagias2013cops}.

\bibliographystyle{amsplain}

\end{document}